\documentclass[12pt, onehalfspacing]{article} 
\usepackage[T1]{fontenc}

\usepackage{libertine}

\usepackage{geometry} 
\usepackage{xcolor, graphicx} 
\usepackage{tikz}
\usepackage{multicol}
\usetikzlibrary{patterns}

\usepackage{booktabs} 
\usepackage{array} 
\usepackage{paralist} 
\usepackage{verbatim} 
\usepackage{subfig} 
\usepackage{amsfonts, amsmath, amssymb, amsthm, euscript, dsfont, bbm}
\usepackage{mathrsfs}
\usepackage{mhsetup, mathtools}
\usepackage{thmtools}
\usepackage{halloweenmath}
\usepackage{tikz}
\usetikzlibrary{shapes.geometric,backgrounds}

\newtheorem{proposition}{Proposition}  
\newtheorem{lemma}{Lemma}

\newtheorem{theorem}{Theorem}
\newtheorem{corollary}{Corollary}

\theoremstyle{definition}
\newtheorem{definition}{Definition}
\newtheorem{axiom}{Axiom}
\newtheorem*{axiom*}{Axiom}
\newtheorem{example}{Example}

\renewcommand\thmcontinues[1]{Continued}

\newcommand{\ra}{\rightarrow}

\newcommand{\RA}{\Rightarrow}

\DeclareMathOperator*{\argmax}{arg\,max}

\definecolor{purple}{RGB}{85, 6,139}
\definecolor{teal}{RGB}{2,108,128}
\definecolor{lavender}{RGB}{129, 102, 122}
\definecolor{carolina blue}{RGB}{68, 157, 209}
\definecolor{phthalo blue}{RGB}{2, 8, 135}
\definecolor{purple2}{RGB}{149, 96, 219}
\definecolor{green1}{RGB}{96, 219, 117}
\definecolor{orange1}{RGB}{208,70,0}

\usepackage{natbib}
\usepackage{hyperref}
\usepackage{nameref}
\hypersetup{
	breaklinks=true,
  colorlinks   = true, 
  urlcolor     = purple, 
  linkcolor    = teal, 
  citecolor   = gray
}

\usepackage{fancyhdr} 
\usepackage{sectsty}
\usepackage{setspace}
\allsectionsfont{\upshape \raggedright \fontsize{13}{15} \selectfont} 

\usepackage[nottoc,notlof,notlot]{tocbibind} 
\usepackage[titles,subfigure]{tocloft} 

\title{\textsc{Reference Dependence and Random Attention}\thanks{This paper combines the major results from \cite{kovach2016thinking} and \cite{suleymanov2018}. We would like to thank Federico Echenique, Shaowei Ke, Jay Lu, Yusufcan Masatlioglu, Romans Pancs, Pietro Ortoleva, Collin Raymond, and Levent \"Ulk\"u for helpful comments and suggestions. We would also like to thank the audiences of the spring 2016 Midwest Economic Theory Conference, the 2016 North American Summer Meeting of the Econometric Society,  RUD 2016, Foundations of Utility and Risk Conference 2018, and the 2018 CIREQ Montr\'eal Symposium in Microeconomic Theory. Lastly, we thank the editor and three anonymous reviewers for their comments, which led to substantial improvements in the paper. All errors are ours.}}  
\author{\href{https://www.matthewkovach.com/}{Matthew Kovach}\footnote{Department of Economics, Virginia Tech.  E-mail: mkovach@vt.edu.} \and \href{https://sites.google.com/view/elchin-suleymanov/}{Elchin Suleymanov}\footnote{Department of Economics, Purdue. E-mail: esuleyma@purdue.edu.}}

\date{\today}

\begin{document}
\maketitle

\vspace{7 mm}

\begin{abstract}We explore the ways that a reference point may direct attention. Utilizing a stochastic choice framework, we provide behavioral foundations for the Reference-Dependent Random Attention Model (RD-RAM). Our characterization result shows that preferences may be uniquely identified even when the attention process depends arbitrarily on both the menu and the reference point. We then analyze specific attention processes, characterizing reference-dependent versions of several prominent models of stochastic consideration. Our analysis illustrates which attention processes can accommodate behavioral patterns commonly observed in studies, such as frequency reversals among non-status quo alternatives and choice overload.\\


\end{abstract}
\vspace{7 mm}



\vspace{30 mm}

\pagebreak

\section{Introduction}

Decision makers (DM) frequently make choices in the presence of a status quo.  Some examples include an employee selecting among investment options for their $401$k plan after the company has designated a default option, an employee contemplating a job change, or a consumer choosing among various options with their typical purchase in mind.  Further, many economics and psychology experiments have illustrated that the existence of a status quo significantly impacts the distribution of choices. Consequently, numerous models of status quo bias and reference-dependent choice have proliferated.\footnote{For example, see \cite{kahneman1979prospect}, \cite{tversky1991loss},  \cite{munro2003theory}, \cite{sugden2003reference}, \cite{apesteguia2009theory}, \cite{masatlioglu2005rational, masatlioglu2014canonical}, \cite{sagi2006anchored}.}

However, some striking behavioral patterns are difficult to accommodate with existing models. First, changing the status quo may induce frequency reversals among non-status quo alternatives \citep{masatlioglu2013reference, dean2017limited}.  Second, a status quo may shift attention towards a particular category of goods, rather than merely draw attention towards itself \citep{Maltz2020category}. Third, a DM may be more likely to choose the status quo in large choice sets \citep{iyengar2000choice, Dean2008}. These examples suggest that the status quo plays a role in directing the DM's attention.

To accommodate these and other behavioral patterns, we introduce an observable status quo to the stochastic choice framework and provide a behavioral analysis of \emph{reference-dependent attention}.\footnote{There are many situations in which the reference might be observable (e.g., current job, investments, among many others) and observability of the reference has been assumed in much of the literature \citep{masatlioglu2005rational, Dean2008, masatlioglu2013reference, masatlioglu2014canonical, dean2017limited, guney2018costly}.
}  The primitive in our analysis is a family of reference-dependent random choice rules $\{p_r\}_{r \in X}$, where each $p_r$ encodes the DM's stochastic choices from various menus when $r$ is the reference alternative.

We begin with a flexible model that allows for a general attention process, which we refer to as the \textbf{Reference-Dependent Random Attention Model} (RD-RAM).  In the RD-RAM, when the DM faces a menu $S$ and has a reference alternative $r\in S$, a random consideration set $D$ is realized according to the reference-dependent attention rule $\mu_r(\cdot,S)$. Then she maximizes a reference-independent preference, $\succeq$, over the realized consideration set. 

In its most general form, the only restriction on $\mu_r$ is that  $\mu_r(D,S)>0$ only if $r \in D \subseteq S$.\footnote{In much of the paper we focus on $\mu_r$ that satisfy full support:   $\mu_r(D,S)>0$ for all admissible $D$. We provide characterizations with and without this property.}  This restriction entails two meaningful channels through which the reference influences behavior. First, the presence of this restriction ensures that $r$ is always considered, and so in the RD-RAM the reference is \emph{attention privileged}. Second, the absence of any other restriction permits rich interactions among the reference, the menu, and the DM's attention, and so in the RD-RAM the reference \emph{directs attention}.

%
%
%

While the RD-RAM is general, our main result establishes that it does have empirical content and that a unique preference is identifiable from choice data. We characterize the RD-RAM with three novel axioms. The first axiom, \nameref{NCC} (NCC), is a weak notion of status quo bias for stochastic data. It ensures that the revealed preference is acyclic. The second axiom, \nameref{SQA}, requires that if $x$ is never chosen when $y$ is the status quo, then $y$ should sometimes be chosen when $x$ is the status quo. The third axiom, \nameref{NRE}, requires that the reference alternative is always chosen with positive probability.  These three axioms link behavior across references, enable the identification of a reference-independent preference, and ensure that the consideration set always includes the reference.

The flexibility of $\mu_r$ enables the RD-RAM to serve as a framework to compare more structured models of the attention process. In particular, the RD-RAM nests reference-dependent versions of three well-known random attention models: Independent Random Attention (IRA) \citep{manzini2014stochastic},  Luce Random Attention (LRA) \citep{brady2016menu}, and Constant Random Attention (CRA) \citep{aguiar2017random}. Consequently, each of these models share the common structure of the RD-RAM. We establish novel characterizations for IRA (\autoref{theorem:IRA}), LRA (\autoref{theorem:logit}), and CRA (\autoref{theorem:CRA}) in this new domain, which we view as our primary contribution.  Our characterization utilizes \autoref{proposition2}, which shows that IRA is the precise intersection of CRA and LRA. 

Starting from the RD-RAM, we characterize reference-dependent IRA with two additional axioms: \nameref{IDA} and \nameref{RIDA}. The first axiom requires that the choice probabilities of dominant alternatives do not change when a dominated alternative is removed from a menu. The second axiom requires that the relative choice probability for two alternatives stays the same when the dominant alternative in a menu is removed. Each of these axioms places a restriction on how the distribution of consideration sets changes across menus.  We then provide characterizations for LRA and CRA, each of which involves relaxing the axioms of IRA.  To characterize LRA, we retain \nameref{RIDA} but replace \nameref{IDA} with a condition ensuring that the odds of choosing the reference decrease as the menu grows. To characterize CRA, we retain \nameref{IDA} but replace \nameref{RIDA} with a condition ensuring that the probability of choosing the reference decreases as the menu grows. One benefit of our framework with observable references is that, unlike the original characterizations of these models, our characterizations do not require one to specify an observable outside option. To clarify, since the reference point is always available, it plays the same role in our analysis as the outside option in previous characterizations.

The RD-RAM and the special cases we consider capture a variety of rich behavioral patterns. Since the reference directs attention, changes in the reference can cause large changes in choice probability by shifting attention. For example, IRA may capture category effects and frequency reversals among non-status quo alternatives, both of which are regularly observed experimentally. Another empirical pattern that is well-documented is choice overload, or the increasing tendency to select the default when more options are available. The LRA allows for choice overload, while IRA and CRA do not. Moving beyond these special cases, the RD-RAM can allow for even richer behavior, such as salience-based attention. These examples are discussed in \autoref{section: discussion}.

A natural question for any model of reference-dependence is whether ``being the reference'' unambiguously improves an alternative's chances of selection.  To answer this, we introduce \nameref{SQM}, a new condition formalizing this idea. This condition states that, in any menu, an alternative is most likely to be chosen when it is the reference alternative. Recall that in the RD-RAM there are two channels of influence for the reference; references are attention privileged and direct attention.  References being attention privileged increases the likelihood of choosing the reference.  References directing attention has an ambiguous effect, since the reference may direct attention towards preferable alternatives.  Therefore, the net effect of the two channels is also ambiguous. We show that even the most restrictive IRA model (and hence CRA and LRA) may violate \nameref{SQM} due to ``shifts'' in attention; if $x$ becomes the status quo, this may shift attention towards preferred alternatives, which increases the chances of abandoning the status quo. 


Intuition suggests that  \nameref{SQM} might be satisfied if we restrict the references's ability to direct attention, which would require eliminating the second channel of influence for the reference. To investigate this, we define a notion of reference-independent attention for the IRA, LRA, and CRA models. In each of these models, the reference retains its privileged status but does not otherwise influence consideration probabilities. We show that the reference-independent CRA always satisfies \nameref{SQM}, hence so does the reference-independent IRA (\autoref{proposition: SQM - CRA}). However, \nameref{SQM} may still fail for the reference-independent LRA.  

The remainder of the paper is structured as follows. In \autoref{section: model} we introduce the model, characterize the RD-RAM, and discuss the relationship between IRA, LRA, and CRA models. The characterizations of the general model, as well as the IRA, LRA, and CRA models, appear in \autoref{section: characterization}. Our discussion on behavioral implications of the model, the analysis of \nameref{SQM}, and the relationships between the RD-RAM and other well-known models are in \autoref{section: discussion}. The proofs of all results are collected in \autoref{section: proofs}.

\section{Model}\label{section: model}

\subsection{Reference-Dependent Random Attention}

This section defines random choice rules for problems with an observable status quo. Let $X$ be a finite set of alternatives and $\mathcal{X}$ be the collection of non-empty subsets of $X$ (menus).  The collection of menus that contain $r$ is denoted by $\mathcal{X}_r$. A choice problem is a pair $(S,r)$ where $S\in \mathcal{X}_r$. Our observable primitive is the family of reference-dependent random choice rules indexed by each potential status quo: $\{p_r\}_{r \in X}$.

\begin{definition} A \textbf{reference-dependent random choice rule} is a map $p_{r}:X \times \mathcal{X}_r \ra [0,1]$ such that for all problems $S\in \mathcal{X}_r$,  $p_r(x,S)>0$ only if $x\in S$ and $\sum_{x \in S}p_r(x,S)=1$.
\end{definition}

We begin our analysis by introducing a general model of ``unstructured'' reference-dependent consideration that we call the \textbf{Reference-Dependent Random Attention Model} (RD-RAM). This establishes a common framework that enables us to characterize various reference-dependent consideration models. Further, this approach allows us to clarify the distinguishing features among these models.

In the RD-RAM, each reference alternative has an associated \textit{random attention rule} $\mu_r: \mathcal{X}\times \mathcal{X}_r\rightarrow [0,1]$ such that $\mu_r(D,S)>0$ only if $r\in D\subseteq S$ and $\sum_{D:\: D\subseteq S}\mu_r(D,S)=1$. That is, $\mu_r(D,S)$ denotes the probability that the DM's consideration set is $D\subseteq S$ when the choice problem is $(S,r)$. We say that a random attention rule has \textit{full support} if $\mu_r(D,S)>0$ whenever $r\in D\subseteq S$ for any $(S,r)$. We will focus on random attention rules that have full support.\footnote{This is a common assumption in the literature, and it is satisfied in all the special cases that we consider in this paper. Note that this does not require that choice frequencies are always positive, as choice depends on the realized consideration set and the DM's preference. However, one can still imagine that individual consideration may be sparse. To address this possibility, we also provide a characterization of RD-RAM that does not assume full support. Please see the discussion after Theorem \ref{theorem:RDRAM} for more details.} In each choice problem $(S,r)$, the DM realizes some consideration set, drawn from $\mu_r(\cdot, S)$, and then maximizes her preference $\succeq$, which we assume to be a linear order.\footnote{A binary relation $\succeq$ on $X$ is a linear order if it is (i) complete: for any $x,y\in X$, either $x\succeq y$ or $y\succeq x$; (ii) antisymmetric: for any $x,y\in X$, $x\succeq y$ and $y\succeq x$ imply $x=y$; and (iii) transitive: for any $x,y,z\in X$, $x\succeq y\succeq z$ implies $x\succeq z$.}

\begin{definition}
\label{definition:rdram}
Reference-dependent random choice $\{p_r\}_{r\in X}$ has a random attention representation (alternatively, $\{p_r\}_{r\in X}$ is an \textit{RD-RAM}) if there exist a linear order $\succeq$ on $X$ and a full support random attention rule $\{\mu_r\}_{r\in X}$ such that for any choice problem $(S,r)$ and for any $x\in S$, 
\begin{equation}
p_r(x,S)=\underset{D:\: r\in D\subseteq S,\: x=\argmax(\succeq, D)}{\sum}{\mu_r(D,S)}.
\end{equation}
\end{definition}

We can study specific features of attention by introducing structure to $\mu_r$. For instance, we can impose general properties such as the monotonicity condition introduced by \cite{cattaneo2020random}, which requires that the probability of a consideration set may only increase as the number of possible consideration sets decreases: for any $r\in D\subseteq S\subseteq T$, $\mu_r(D,S)\geq \mu_r(D,T)$. Alternatively, we can consider more structured attention processes. For example, the RD-RAM nests the models of Independent Random Attention (IRA) \citep{manzini2014stochastic}, Luce Random Attention (LRA)  \citep{brady2016menu}, and Constant Random Attention (CRA) \citep{aguiar2017random}, which we formally define below. Each of these special cases also satisfies the monotonicity condition of  \cite{cattaneo2020random}. We provide novel foundations for each of these models in \autoref{section: characterization}.

\begin{definition}\label{definition:attnmodels}
Let $\{p_r\}_{r\in X}$ have an RD-RAM representation $(\succeq, \{\mu_r\}_{r\in X})$. We say that $\{p_r\}_{r\in X}$ has 

\begin{enumerate}
\item an \textit{Independent Random Attention (IRA)} representation if there exist attention maps $\gamma_r:X \rightarrow [0,1]$ satisfying $\gamma_{r}(x) \in (0,1)$ and $\gamma_{r}(r)=1$, and for any $r\in D\subseteq S$,
	\begin{equation}
	\label{attn:ira}
	\mu_r(D,S)=\prod_{x\in D}\gamma_r(x)\prod_{y\in S\setminus D}(1-\gamma_r(y)).
	\end{equation}

	\item a \textit{Luce Random Attention (LRA)} representation if there exist
	 $\pi_r:\mathcal{X}\rightarrow [0,1]$ such that $\pi_r(D)>0$ whenever $r\in D\subseteq X$, $\sum_{D:\: r \in D\subseteq X}\pi_r(D)=1$, and for any $r\in D\subseteq S$,
	\begin{equation}
	\label{attn:lra}
	\mu_r(D,S)=\frac{\pi_r(D)}{\sum_{D':\: r \in D'\subseteq S}\pi_r(D')}.
	\end{equation}

	\item a \textit{Constant Random Attention (CRA)} representation if there exist $\pi'_r:\mathcal{X}\rightarrow [0,1]$ such that $\pi'_r(D)>0$ whenever $r\in D\subseteq X$, $\sum_{D:\: r \in D\subseteq X}\pi'_r(D)=1$, and for any $r\in D\subseteq S$,
	\begin{equation}
	\label{attn:cra}
	\mu_r(D,S)=\sum_{D': D'\cap S=D}\pi'_r(D').
	\end{equation} 
\end{enumerate}
\end{definition}	

One benefit of our framework is that our characterizations of these models do not require one to specify an observable outside option. This is in contrast to previous characterizations with the exception of \cite{horan2019random}, which provides a characterization for the IRA rule without assuming an observable outside option. To elaborate, since the reference point is always available to the DM, it plays the same role in our analysis as the outside option in previous characterizations.

A common thread running through each of these special cases is that menu-dependent attention probabilities are determined by a menu-independent primitive. One advantage of this approach is that it increases the applicability of the RD-RAM by imposing more structure on observed choices. Additionally, this also allows us to identify the attention primitives from observed choices. This is not possible in general when arbitrary forms of menu-dependence are assumed. Hence, understanding the behavioral implications of these special cases in this new setup is crucial. Our analysis helps us understand which of these special cases can accommodate the types of behaviors commonly observed in studies. For example, we show that while the IRA rule (and hence both LRA and CRA rules) can capture choice frequency reversals, only the LRA rule can accommodate choice overload (see \autoref{behavior}). 
	
We note that with arbitrary menu-dependence, even the most restrictive consideration set models can become too general. To illustrate, consider the IRA rule, which is the most restrictive model among the special cases we consider. It can be shown that a menu-dependent extension of IRA, where $\gamma_r(x)$ depends arbitrarily on $S$, is equivalent to the RD-RAM. The advantage of the RD-RAM is that it incorporates all these special cases and provides a unified framework for our study.

\subsection{Relationship Between Models}	\label{attentionmodels}

Our first main result shows that Independent Random Attention is equivalently characterized as the intersection of Luce Random Attention and Constant Random Attention. It is well-known that both LRA and CRA representations are generalizations of IRA (see \cite{brady2016menu,aguiar2017random}). Accordingly, each of these generalizations features an attention rule that relaxes \autoref{attn:ira}. 
For LRA, the key feature is an Independence of Irrelevant Alternatives (IIA) property for consideration sets. For any $S,T\in \mathcal{X}$, $D,D'\subseteq  S\cap T$ and $r\in D\cap D'$,
\[\frac{\mu_r(D,S)}{\mu_r(D',S)}=\frac{\pi_r(D)}{\pi_r(D')}=\frac{\mu_r(D,T)}{\mu_r(D',T)}.\]
For CRA, the key feature is that each alternative is considered with the same probability for all choice sets in which it is available. For any $S,T\in \mathcal{X}$ and $r,x\in S\cap T$, 
\[\sum_{D: x\in D}\mu_r(D,S)=\sum_{D: x\in D}\pi_r(D)=\sum_{D: x\in D}\mu_r(D,T).\]
Importantly, we show that if a reference-dependent random attention rule $\mu_r$ has both LRA and CRA representations, then it must also have an IRA representation. 

\begin{proposition}
	\label{proposition2}
	A random attention rule $\mu_r$ has an IRA representation if and only if it has LRA and CRA representations. 
\end{proposition}

Note that \autoref{proposition2} is stated in terms of the attention rule, which is not observed. In general, if the observed choice rule $\{p_r\}_{r\in X}$ has an LRA representation $(\succeq,\{\pi_r\}_{r\in X})$ and a CRA representation $(\succeq',\{\pi'_r\}_{r\in X})$, \autoref{proposition2} on its own does not imply that $\{p_r\}_{r\in X}$ must have an IRA representation. This is because the preferences, $\succeq$ and $\succeq'$, in the representations might be different. However, we later show in \autoref{theorem:RDRAM} that preference is unique in any RD-RAM representation. Hence, the result in \autoref{proposition2} extends to the observed choice rule $\{p_r\}_{r\in X}$. The following corollary, therefore follows directly from \autoref{theorem:RDRAM} and \autoref{proposition2}. 

\begin{corollary}
	\label{corollary1}
	A reference-dependent random choice rule $\{p_r\}_{r\in X}$ has an IRA representation if and only if it has LRA and CRA representations.
\end{corollary}


To summarize, we have shown that a reference-dependent random choice (attention) rule has an IRA representation of and only if it has LRA and CRA representations. \autoref{fig:models} illustrates the result in this section. 

\begin{figure}[h!]
	\centering
%
\includegraphics[width=8cm]{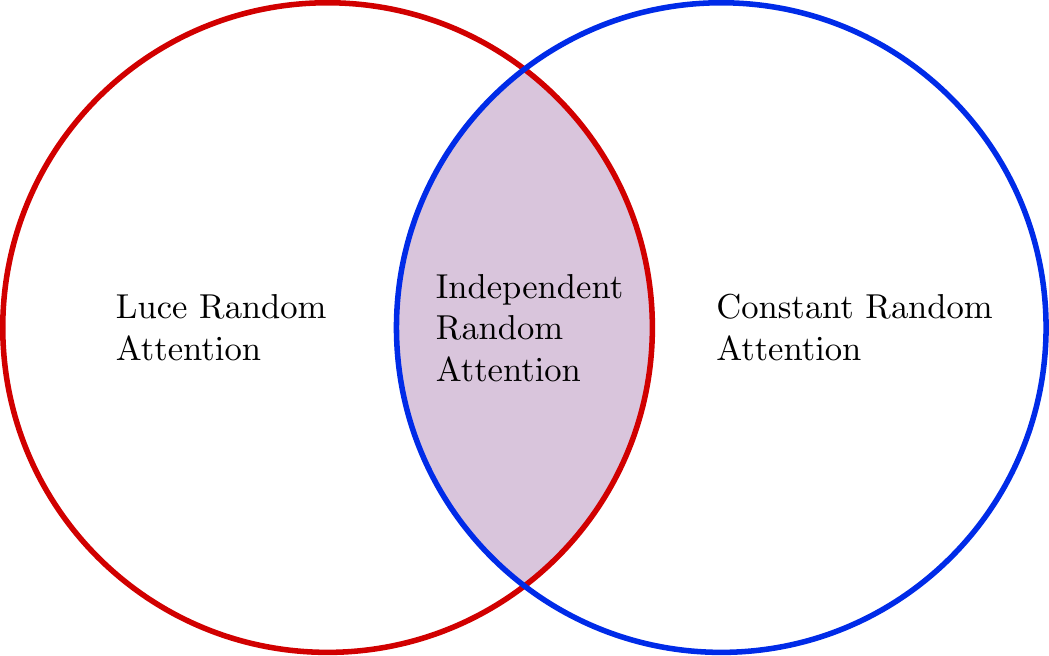}
	\caption{Relationship between models.}
	\label{fig:models}
\end{figure}

\section{Characterization}\label{section: characterization}

Our characterization of RD-RAM rests upon a novel behavioral postulate linking choices across different references. For an intuition behind this postulate, imagine a job seeker currently employed in New York. Further, suppose that in spite of any status quo bias she may feel, she would be willing to relocate to California.  In the complementary situation in which she were already employed in California, it is natural to suppose that she would refuse to accept a job in New York. This imposes a natural asymmetry in a DM's willingness to abandon a status quo. In addition, suppose the DM is also willing to relocate from Michigan to New York. Given that she is willing to relocate from New York to California, it is natural to assume that she would not be willing to relocate from California to Michigan. Our key axiom, \nameref{NCC}, captures this intuition and ensures that the decision maker's choices do not exhibit ``frequency cycles" as in this example.

\begin{axiom}[No-Cycle Condition]\label{NCC}
	For any collection $(x_i,S_i)_{i=1}^n$, if $p_{x_i}(x_{i+1},S_i)>0$ for $i<n$, then $p_{x_n}(x_1,S_n)=0$. 
\end{axiom}


For additional insight into our axiom, it is instructive to compare it with the axiom of status quo bias from \citet{masatlioglu2005rational}.  For a choice correspondence $C$, they impose the following condition: for any $(S,y)$, if $x \in C(S,y)$ then $\{x\} = C(S,x)$. The intuition behind this is that when an alternative $x$ is chosen over the status quo $y$ in some choice problem, then when $x$ becomes the status quo and the set of alternatives is unchanged, it must be the unique choice due to status quo bias.

The direct analogue of this condition would require that whenever $x$ is chosen with positive probability in problem $(S,y)$, then it must be chosen from $(S,x)$ with probability one: $p_y(x,S)>0 \implies p_x(x,S)=1$. On the other hand, \nameref{NCC} only requires that $y$ is never chosen when $x$ is the staus quo, i.e., $p_x(y,S)=0$. This is because under full attention, if $x \in C(S,y)$ we may infer that $x$ (i) dominates the status quo and (ii) is no worse than any other alternative in $S$ that also dominates the status quo.  Under random attention we cannot reason this way, because a DM may choose $x$ only because she (stochastically) failed to consider a third alternative $z$ that dominates $x$. In this case, it is plausible that she may choose $z$ with positive probability even if $x$ were the status quo.  

Our next axiom, \nameref{SQA}, requires that if $x$ is never chosen when $y$ is the status quo, then $y$ must always be chosen with positive probability when $x$ is the status quo. An implication of this property is that reference effects are not too strong to prevent better alternatives from being chosen. To get an intuition behind this axiom, suppose the job seeker is not willing to relocate to New York when she is currently employed in California. \nameref{SQA} then requires that she should be willing to relocate from New York to California with positive probability.

\begin{axiom}[Status Quo Asymmetry]\label{SQA}  For all $x,y \in S\cap T$, $x\neq y$,
	\[p_y(x,S)=0 \quad \RA\quad p_x(y,T)>0 \]
\end{axiom}




Our last axiom states that the reference alternative is always chosen with positive probability. Hence, the job seeker in our example should stay in New York with positive probability if she is currently employed in New York and she should stay in California with positive probability if she is currently employed in California. They key behind this axiom is the intuition that the status quo is \emph{attention privileged}, and hence always considered. The full support assumption then guarantees that there is some positive probability with which the DM fails to consider any alternative more preferred to the status quo. 

\begin{axiom}[Nontrivial Reference Effect]\label{NRE}
	For any $x\in S$, we have $p_x(x,S)>0$. 
\end{axiom}

The following theorem shows that \nameref{NCC}, \nameref{SQA}, and \nameref{NRE} are necessary and sufficient for an RD-RAM representation. It also states that preference is uniquely revealed in any RD-RAM.

\begin{theorem}
	\label{theorem:RDRAM}
	$\{p_r\}_{r\in X}$ has a RD-RAM representation if and only if it satisfies \nameref{NCC}, \nameref{SQA}, and \nameref{NRE}. Moreover, if $(\succeq,\{\mu_r\}_{r\in X})$ and $(\succeq',\{\mu'_r\}_{r\in X})$ are two RD-RAM representations of $\{p_r\}_{r\in X}$, then $\succeq=\succeq'$.
\end{theorem}

As previously mentioned, we focus on the full-support RD-RAM throughout this paper as it is implied by all the special cases we consider and it simplifies their exposition. While the full-support RD-RAM may approximate the RD-RAM without full support, the condition is not entirely innocuous. For example, imagine a job seeker who would refuse to accept any job in New York if she were currently employed in California and also would refuse to accept any job in California if she were currently employed in New York.  This suggests that $p_{x_{ca}}(y_{ny},S)=0=p_{y_{ny}}(x_{ca},S)$, a violation of \nameref{SQA}. Similarly, a DM currently employed in Michigan may not want to stay in Michigan at all and may prefer to move to other cities with probability one. This would be in violation of \nameref{NRE}. However, relaxing the full-support assumption in our general model is not too difficult. In fact, \nameref{NCC} fully captures all the behavioral implications of the RD-RAM without full support. 

One possible justification for the full support assumption would be to view the observed choice data as the one generated by a population of agents which share the same preference relation but display heterogeneity in attention. In this case, one can view the full support assumption as equivalent to assuming rich heterogeneity of attention. Of course, assuming that all the agents share the same preference relation would be restrictive in itself. To address this issue, one can assume heterogeneity both in preferences and in attention, as in \cite{kashaev2022random}. However, allowing for heterogeneity both in preferences and in attention is not as straightforward in our setup. One difficulty arises from the fact that once we allow for heterogeneity in preferences, one leads to the assumption that the distribution of preferences depends on the reference point too. However, a general model like this would impose few behavioral restrictions without further assumptions on how the distributions of random preferences across different references are related.


One advantage of the full support assumption is that, as stated in \autoref{theorem:RDRAM}, preference is fully revealed. To see why, notice that we must have $x\succeq y$ if $p_y(x,T)>0$ for some $T$ regardless of the full support assumption. In addition, if $x\succeq y$, then the full support assumption ensures that $p_y(x,T)>0$ for all $T\supseteq \{x,y\}$. The latter ensures that preference is fully revealed. Since the three special cases we consider, the IRA, LRA, and CRA rules, satisfy the full support assumption, our result can be applied to each of these attention rules. 

If the full support assumption is violated, then preference can still be partially revealed from observed choices by letting $x\succeq y$ if $p_y(x,T)>0$ for some $T\supseteq \{x,y\}$. However, without the full support assumption, $\succeq$ may no longer be complete. To infer additional information about preferences, it is then necessary to impose more restrictions on the attention rules. For example, if the attention rules satisfy the monotonicity property of \citet{cattaneo2020random}, then the observation $p_r(x,S)>p_r(x,S\setminus y)$ reveals that $x$ must be preferred to $y$.\footnote{If the attention rules in the RD-RAM also satisfy the monotonicity property of \citet{cattaneo2020random}, then the observed choice rule will satisfy the following weak regularity property: if $p_x(y,T)>0$ for some $T\in \mathcal{X}$, then $p_r(x,S)\leq p_r(x,S\setminus y)$ for all $(S,r)$. On the other hand, if the observed choice rule satisfies the weak regularity property in addition to \nameref{NCC}, \nameref{SQA}, and \nameref{NRE}, then one can construct an RD-RAM representation where the attention rules satisfy the monotonicity property. Hence, weak regularity captures the behavioral implications of the monotonicity property on attention rules.}

In the remainder of this paper, we will use the following terminology. We say that $x$ is the \textit{dominant alternative} in $S$ if $x$ is never abandoned in $S$ when it is the status quo alternative, i.e., $p_x(x,S)=1$. If $\{p_r\}_{r\in X}$ has an RD-RAM representation, then there can be only one dominant alternative in any set. To see this, notice that $p_x(y,S)=0$ implies $p_y(x,S)>0$ by \nameref{SQA}, and hence $p_y(y,S)\neq 1$. In addition, since $p_y(x,S)>0$ only if $x\succ y$, $p_x(x,S)=1$ reveals us that $x$ is $\succeq$-maximal element in $S$.

\subsection{Independent Random Attention}
\label{section: IRA}

In this section we provide a characterization of the IRA model in terms of observed reference-dependent random choices $\{p_r\}_{r\in X}$. An agent whose attention is described by \autoref{attn:ira} exhibits random attention \`a la \citet{manzini2014stochastic} with two additional features. First, the status quo must always be considered, and therefore consideration sets are always non-empty.  This also dispenses with the need for the outside option utilized by  \citet{manzini2014stochastic}. Second, the status quo directs attention via status quo dependent attention mappings, $\gamma_{r}$. Finally, our characterization is quite different from the one established by \citet{manzini2014stochastic} and so is of independent interest. 

\nameref{NCC}, \nameref{SQA}, and \nameref{NRE} are necessary for an IRA representation, since $\{p_r\}_{r\in X}$ is an RD-RAM. We now provide two additional axioms that characterize the IRA representation. The first axiom, Irrelevance of Dominated Alternatives (IDA), states that if $x$ is chosen with a positive probability in some choice set when $y$ is the reference point (and hence $x$ "dominates" $y$), then, regardless of the reference point, removing $y$ from any choice set cannot affect the probability that $x$ is chosen. In other words, dominated alternatives cannot affect the choice probability of dominant alternatives.

\begin{axiom}[Irrelevance of Dominated Alternatives]\label{IDA} For any $(S,r)$ and $x,y\in S$ such that $y\notin \{ r,x\}$, if $p_y(x,T)>0$ for some $T\in \mathcal{X}$, then $$p_r(x,S)=p_r(x,S\setminus y).$$ 
\end{axiom}

To see why IDA is necessary for an IRA representation, notice that under IRA, for all $y \notin D$,
\[\mu_r(D,S\setminus y)=\mu_r(D\cup y,S)+\mu_r(D,S).\]
That is, the probability that $D$ is the consideration set in the choice set $S\setminus y$ is equal to the probability that either $D$ or $D\cup y$ is the consideration set in $S$. If $x$ dominates $y$, then $x$ is the best alternative in $D\subseteq S\setminus y$ if and only if it is the best in $D\cup y$. This together with the above equality imply that removing $y$ from $S$ cannot affect the probability that $x$ is chosen. 

The next property relates relative choice probabilities of two alternatives when the dominant alternative in the choice set is removed from it. It says that if $x$ is the dominant alternative in $S$, then for any two alternatives $y,z\in S\setminus x$ which are chosen with positive probability, the relative probability that $y$ is chosen as opposed to $z$ is unaffected by the removal of $x$. This is a weakening of Luce's Independence of Irrelevant Alternatives (IIA) axiom \citep{luce1959individual}. Recall that Luce's IIA states that this ratio stays the same when we compare arbitrary choice sets $S$ and $T$ that include $y$ and $z$. While this does not hold in general for the IRA model, the dominant alternative in a menu can be considered irrelevant for the relative probabilities of other alternatives. 

\begin{axiom}[Ratio Independence of Dominant Alternatives]\label{RIDA} For any $(S,r)$, $x\in S\setminus r$, and $y,z\in S\setminus x$, if $p_x(x,S)=1$ and $p_r(y,S)p_r(z,S)>0$, then
	$$\frac{p_r(y,S)}{p_r(z,S)}=\frac{p_r(y,S\setminus x)}{p_r(z,S\setminus x)}.$$
\end{axiom}

To see why Ratio Independence of Dominant Alternatives (RIDA) holds in the IRA model, notice that since $x$ is the dominant alternative in $S$, removing it from the menu does not affect which consideration sets have $y$ or $z$ as the best alternative; if $y$ (or $z$) is the best alternative in $D\subseteq S$, then $x$ cannot be in $D$ so that $D\subseteq S\setminus x$. In addition, independent random attention implies that for any $D,D'\subseteq S\setminus x$,
\[\frac{\mu_r(D,S)}{\mu_r(D',S)}=\frac{\mu_r(D,S\setminus x)}{\mu_r(D',S\setminus x)}.\]
This equation, combined with the above observation, imply that RIDA is necessary for an IRA representation. 

We are now ready to provide a characterization result for the IRA model. 

\begin{theorem}
	\label{theorem:IRA}
	Assume that $\{p_r\}_{r\in X}$ is an RD-RAM. Then $\{p_r\}_{r\in X}$ has an Independent Random Attention representation if and only if it satisfies \nameref{IDA} and \nameref{RIDA}. 
\end{theorem}

Here we sketch the main steps of the proof (full details are in \autoref{section: proofs}). First, an implication of RIDA is that for any $(S,r)$, if $x$ is the dominant alternative in $S$, the ratio $p_r(y,S)/p_r(y,S\setminus x)$ is constant for all $y\in S$ with $p_r(y,S)>0$. Then, IDA implies that this constant is equal to $p_r(r,\{x,r\})$, which is strictly positive by \nameref{NRE}. Next, for any $r,z\in X$, we define $\gamma_r(z)$ to be the maximum of $p_r(z,\{z,r\})$ and $p_z(r,\{z,r\})$. The proof then proceeds by showing the representation for binary choice sets and using induction together with the previous conclusion to extend the representation to all choice sets.

In terms of uniqueness, preference is uniquely identified from observed choices in the IRA model, since it is a special case of the RD-RAM. Attention maps are uniquely identified in the upper contour set of reference alternatives. Since the value of $\gamma_r(z)$ does not affect choices when $r\succ z$, attention maps are not identified in the lower contour set of reference alternatives.  \autoref{theorem:RDRAM} and \autoref{theorem:IRA} imply the following corollary.

\begin{corollary}
	\label{corollary2}
	Assume that $\{p_r\}_{r\in X}$ has an Independent Random Attention representation. Then, if $(\succeq,\{\gamma_r\}_{r\in X})$ and $(\succeq',\{\gamma'_r\}_{r\in X})$ are two IRA representations of $\{p_r\}_{r\in X}$, we have that $\succeq=\succeq'$ and $\gamma_r(x)=\gamma'_r(x)$ for any $r\in X$ and $x\succeq r$. 
\end{corollary}


\subsection{Luce Random Attention}
\label{section: LRA}

We now provide a characterization for the Luce Random Attention model. As discussed in the previous section, LRA is a generalization of IRA, and it relaxes one of the axioms of IRA. Since the discussion following \autoref{RIDA} also holds for LRA, RIDA is still necessary for an LRA representation. However, LRA no longer satisfies the IDA property. In fact, in the LRA model the choice probability of an alternative may increase or decrease when a dominated alternative is added to a menu.

Our next example illustrates that RIDA by itself is not sufficient for a Luce Random Attention representation while shedding light on the additional axiom that we need for the characterization. 

\begin{example}[Insufficiency of RIDA]
	\label{example: RIDA_logit}
	Let $X=\{x,y,z\}$ be the set of alternatives, and consider a DM represented by some RD-RAM. By  \autoref{theorem:RDRAM}, her preference is uniquely identified from observed choices. Let $x\succ y\succ z$ be the identified preference. Suppose that when the DM's reference point is $z$, 
	her choices satisfy 
	$$p_z(x,\{x,y,z\})<p_z(z,\{x,y,z\}) \text{ and }p_z(a,S)=p_z(z,S) \text{ for any other } a\in S\subseteq X.$$
	In other words, all alternatives are chosen with the same probability as the reference alternative $z$, except for $x$ in $\{x,y,z\}$ where $x$ is chosen with a smaller probability. Since $x$ is the dominant alternative in $\{x,y,z\}$ and $p_z(y,S)=p_z(z,S)$ for $S\in \{\{x,y,z\}, \{y,z\}\}$, the observed choices satisfy RIDA. However, if $\{p_r\}_{r\in X}$ has a Luce Random Attention representation, we cannot have $p_z(x,\{x,y,z\})<p_z(z,\{x,y,z\})$ given the other observations. To see this, note that since $z$ is the worst alternative, for every choice set $S$, $z$ is only chosen when the DM has the degenerate consideration set $\{z\}$. On the other hand, $x$ is chosen from $\{x,y,z\}$ when the consideration sets are either $\{x,y,z\}$ or $\{x,z\}$, and it chosen from $\{x,z\}$ when the consideration set is $\{x,z\}$. This implies that, in the Luce Random Attention model, the relative choice probability of $z$ and $x$ should decrease when $y$ becomes available. The observed choices violate this condition. 
\end{example}

Motivated by this example, we introduce an axiom which captures the intuition that the reference alternative should be less likely to be chosen when alternatives that dominate it become available. Before stating the axiom, we require the following notation. Let $\mathcal{D}_r$ be the collection of subsets of $X$ in which $r$ is the dominant alternative, and let $P_r$ denote the set of alternatives which dominate $r$ (i.e., $s\in P_r \Leftrightarrow r\neq s \text{ and }p_r(s,T)>0 \text{ for some }T\supseteq \{r,s\}$). For any $(S,r)$, the odds that an alternative other than the reference point is chosen (i.e., the odds against the reference alternative) is defined by $\mathcal{O}^{r}_{S}=\frac{1-p_r(r,S)}{p_r(r,S)}$. For any set $U\not\ni r$, let $\Delta_U \mathcal{O}^r_S=\mathcal{O}^{r}_{S}-\mathcal{O}^{r}_{S\setminus U}$ reflect the change in odds against the reference alternative when the choice set is expanded from $S\setminus U$ to $S$. For any collection $\mathcal{U}=\{U_1,\cdots,U_n\}$, let $\Delta_{\mathcal{U}}\mathcal{O}^r_S$ be iteratively defined by $\Delta_{\mathcal{U}}\mathcal{O}^r_S=\Delta_{\mathcal{U}\setminus U_1}\mathcal{O}^r_S-\Delta_{\mathcal{U}\setminus U_1}\mathcal{O}^r_{S\setminus U_1}$. The next axiom has the implication that the odds against the reference alternative increases, or alternatively, the odds for the reference alternative decreases at an increasing rate as the choice set is expanded by introducing more alternatives that dominate the reference alternative. This axiom is the analog of the Increasing Feasible Odds axiom in \cite{brady2016menu} for the LRA model with reference-dependent stochastic choice data. 

\begin{axiom}[Decreasing Odds for the Reference Alternative]
	\label{DORA}
	For any $(S,r)$ and any collection $\mathcal{U}=\{U_1,\cdots,U_n\}$ such that $U_i\subseteq P_r$ and $S\cap U_i\neq \emptyset$ for any $i\in \{1,\dots, n\}$,
	$$\Delta_{\mathcal{U}}\mathcal{O}^r_S>0.$$
\end{axiom}

Notice that \autoref{DORA} is silent on the effect of changing a menu by introducing new alternatives that are dominated by the reference alternative, while IDA requires that the choice probabilities of all alternatives stay the same in this case. In general, while IDA and \autoref{DORA} are independent axioms, from the representation theorem for the IRA model we can conclude that IDA together with RIDA imply \autoref{DORA}. 

We now provide the representation result for the Luce Random Attention model. 

\begin{theorem}
	\label{theorem:logit}
	Assume that $\{p_r\}_{r\in X}$ is an RD-RAM. Then $\{p_r\}_{r\in X}$ has a Luce Random Attention representation if and only if it satisfies \nameref{RIDA} and \nameref{DORA}. 
\end{theorem}

In terms of uniqueness, preference is uniquely identified in any RD-RAM. For the result on the identification of Luce weights, see \cite{suleymanov2018}.

\subsection{Constant Random Attention}
\label{section: CRA}

Our last characterization is for the Constant Random Attention model. Since the discussion following \autoref{IDA} also holds for CRA, IDA is necessary for a CRA representation. However, CRA no longer satisfies RIDA. The next example illustrates that IDA by itself is not sufficient for a Constant Random Attention representation. It also illustrates which additional axiom we need for the characterization.

\begin{example}[Insufficiency of IDA]
	\label{example: IDA_constant}
	Let $X=\{x,y,z\}$ be the set of alternatives, and consider a DM whose choices can be represented by an RD-RAM. Let $x\succ y\succ z$ be the DM's preference. Suppose when the DM's reference point is $z$, her choices satisfy 
	$$p_z(z,\{y,z\})<p_z(z,\{x,y,z\}) \quad \text{ and } \quad p_z(x,\{x,z\})=p_z(x,\{x,y,z\}).$$
	The equality ensures that IDA is satisfied. However, the first observation cannot be consistent with the CRA model. To see this, note that in the CRA model we must have $p_z(z,\{x,y,z\})=\pi_z(\{z\})<\pi_z(\{z\})+\pi_z(\{x,z\})=p_z(z,\{y,z\})$. 
\end{example}

In \autoref{example: IDA_constant}, we have a violation of \textit{regularity} \citep{Suppes-Luce_1965_Handbook}, a condition that states that an alternative should be less likely to be chosen when there are more competing alternatives. As we will discuss later, any choice data consistent with the CRA model must also satisfy regularity. Since IDA imposes a restriction on observed choices only when dominated alternatives are removed from the choice set, it does not impose any restriction on $p_z(z,\{x,y,z\})-p_z(z, \{y,z\})$, and hence it is not sufficient for a CRA representation. In fact, for \autoref{example: IDA_constant} to be consistent with the CRA model, observed choices must satisfy $p_z(z,\{x,y,z\})<p_z(z,\{y,z\})$ and $p_z(y,\{x,y,z\})<p_z(y,\{y,z\})$.

Next, we introduce the axiom that helps us characterize the CRA model. As before, $P_r$ is the set of alternatives which dominate the reference alternative. For any $(S,r)$ and $U\not \ni r$, let $\Delta_{U}p_r(r,S)=p_r(r,S\setminus U)-p_r(r,S)$. This reflects the change in the probability of choosing the reference alternative when the choice set shrinks. For any collection $\mathcal{U}=\{U_1,\cdots,U_n\}$, let $\Delta_{\mathcal{U}}p_r(r,S)$ be defined by $\Delta_{\mathcal{U}}p_r(r,S)=\Delta_{\mathcal{U}\setminus U_1}p_r(r,S\setminus U_1)-\Delta_{\mathcal{U}\setminus U_1}p_r(r,S)$. As \autoref{example: IDA_constant} illustrates, the CRA model implies that $\Delta_{U}p_r(r,S)>0$ if $U\subseteq P_r$ and $U\cap S\neq \emptyset$. In addition, consider $U'\subseteq P_r$ such that $U'\cap S\neq \emptyset$. The CRA model also implies that $\Delta_{U}p_r(r,S\setminus U')>\Delta_{U}p_r(r,S)$ (i.e, $\Delta_{\{U',U\}}p_r(r,S)>0$). That is, when a choice set contains fewer alternatives that dominate the reference alternative, removing another dominating alternative should increase the choice probability of the reference alternative by more. The idea is that when there are fewer alternatives that the reference alternative competes against, removing a dominating alternative from the menu should have a bigger impact on the reference alternative. We formalize this intuition in the following axiom, which is the analog of the axiom in \citet*{aguiar2017random} for the CRA model with reference-dependent stochastic choice data.  

\begin{axiom}[Decreasing Propensity of Choice for the Reference Alternative] 
	\label{DPCRA}
	For any $(S,r)$ and a collection $\mathcal{U}=\{U_1,\cdots,U_n\}$ such that $U_i\subseteq P_r$ and $S\cap U_i\neq \emptyset$ for every $i\in \{1,\dots, n\}$,
	$$\Delta_{\mathcal{U}}p_r(r,S)>0.$$
\end{axiom}

Notice that similar to \autoref{DORA}, \autoref{DPCRA} imposes a structure on observed choices when the choice set is changed by adding or removing alternatives that dominate the reference point. On the other hand, IDA imposes a structure on observed choices when we are adding or removing dominated alternatives. Since any alternative either dominates or is dominated by the reference alternative, IDA together with \autoref{DPCRA} imply regularity for the reference point: for any $S\supseteq T\ni r$, $p_r(r,T)\geq p_r(r,S)$. In fact, they jointly imply regularity for all alternatives, as the representation and the preceding discussion make clear. This is in contrast to the LRA model, which imposes no structure on observed choices when an alternative dominated by the reference point is added to the choice set. Lastly, while RIDA and \autoref{DPCRA} are independent axioms (RIDA by itself does not imply regularity for the reference alternative when an alternative dominating it is removed from the menu), from the representation theorem for the IRA model we can conclude that IDA together with RIDA imply \autoref{DPCRA}. 

We can now provide a characterization result for the CRA model. 

\begin{theorem}
	\label{theorem:CRA}
	Assume that $\{p_r\}_{r\in X}$ is an RD-RAM. Then $\{p_r\}_{r\in X}$ has a Constant Random Attention representation if and only if it satisfies \nameref{IDA} and \nameref{DPCRA}. 
\end{theorem}

As before, preference is uniquely identified in any RD-RAM. For the result on the identification of consideration set probabilities, see \citet{suleymanov2018}.

\autoref{Fig: axioms} illustrates the relationship between all the axioms and summarizes Theorems 2-4. Axiom 4 and 5 jointly characterize the IRA model, which in turn satisfies Axiom 6 and 7. The intersection of Axiom 5 and 6 characterizes the LRA model, while the intersection of Axiom 4 and 7 characterizes the CRA model.

\begin{figure}[h!]
	\centering
	
\includegraphics[width=9cm]{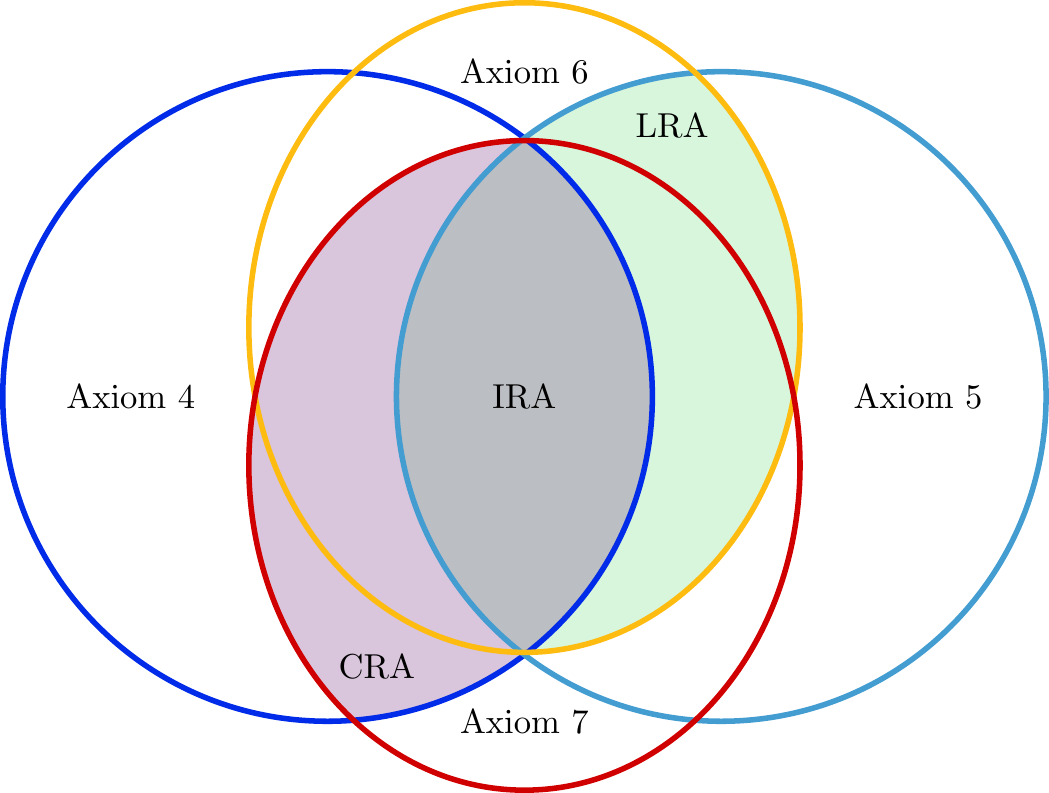}

	\caption{Relationship between axioms.}
	\label{Fig: axioms}
\end{figure}

\section{Discussion}\label{section: discussion}

\subsection{Behavioral Implications}\label{behavior}

\textbf{Generalized Status Quo Bias:} A frequently observed choice pattern involves a frequency reversal between non-status quo alternatives and is referred to as \emph{generalized status quo bias}.  Such reversals occur when an agent chooses $x$ over $y$ when $z$ is the status quo, but the agent chooses $y$ over $x$ when $z'$ is the status quo. As an illustration, consider a traveler selecting a meal upgrade when the default is either a vegetarian or a meat-based meal. Intuitively, the vegetarian default directs her attention toward vegetarian dishes; the meat default directs her attention toward meat dishes. Such choice patterns have been observed in experiments (see \citet{dean2017limited}), but are difficult to explain with most deterministic models of status quo bias.  The IRA rule (and hence CRA and LRA rules) is able to accommodate such reversals by utilizing the idea that the reference point may shift attention between different alternatives. 

\textbf{Category Bias:} A choice pattern that may generate generalized status quo bias is the notion of \emph{category bias}, in which a reference $r$ focuses attention towards goods of the same type. For instance, in choice under uncertainty a DM endowed with a lottery is more likely to choose another lottery than when endowed with a sure payoff.  Category bias can also be captured by the simple IRA rule. For example, suppose $X$ is partitioned into two categories, $X_1$ and $X_2$, where $X(r)$ denotes the category for $r$.  For $0 < \underline{\delta} < \overline{\delta} <1$, let $\gamma_r (x) = \overline{\delta}$ if $x \in X(r)$ and $\gamma_r (x) = \underline{\delta}$ otherwise. It is straightforward to augment this example, say by including alternative-category interactions or more categories, to capture richer attention patterns.  For a deterministic model of category bias and additional discussion, see \cite{Maltz2020category}.  


\textbf{Salience-based Consideration:} One way to add richness to the previous examples is to directly model menu-dependence of attention probabilities. For instance, on crowded store shelves an alternative's color interacts with the colors of the nearby alternatives to direct visual attention. One way to incorporate this richness is to consider a menu-dependent generalization of the IRA rule where $\gamma_r(x)$ depends arbitrarity on $S$. 

As an example, consider the case where the status quo $r$ shifts attention towards goods based on their attributes. In particular, let $X\subseteq \mathbb{R}^N$ where $N$ denotes the number of attributes. Consider a salience function $\sigma: X \times \mathcal{X} \rightarrow \{1, \dots, N\}$ that determines the salient attribute for each potential status quo alternative in any given menu. We can then assume $\gamma_r(x|S) \geq \gamma_r(y|S)$ if $x_{\sigma(r,S)} \geq y_{\sigma(r,S)}$. That is, the DM is more likely to pay attention to $x$ than $y$ if $x$ is ranked higher than $y$ in the salient attribute. More general forms of salience may also be captured through this menu-dependent IRA specification.

\textbf{Choice Overload:} One common observation in the literature on status quo bias is that the agent's likelihood of sticking with the status quo may increase with the number of available alternatives (see \cite{iyengar2000choice} and \citet{Dean2008}).  Accommodating this feature in models of status quo bias has so far been difficult, as nearly all models predict the opposite: adding alternatives should increase the likelihood of abandoning the status quo. One exception is the recent paper by \citet{dean2017limited}, which also combines limited attention with status quo bias. 

Note that choice overload is consistent with the intuition that consideration becomes more taxing in larger menus. That is, when the choice set becomes larger, the probability that the DM considers a particular alternative decreases. Since the reference alternative is attention privileged, this increases the choice probability of the reference alternative in larger choice sets. Since in both IRA and CRA rules an alternative is considered with a fixed probability regardless of the choice set, they cannot capture this choice pattern. On the other hand, in the LRA model, alternatives may be considered with a lower probability in larger choice sets, which allows it to capture choice overload.

\subsection{Status Quo Monotonicity}

While it is now abundantly clear that the status quo may have a significant impact on choice frequencies, one may wonder if an alternative always \emph{benefits} from being the status quo. If $x$ becomes the status quo, must it be chosen with higher probability? This section formalizes this idea by introducing a condition we refer to as \nameref{SQM} (SQM).

\begin{definition}[Status Quo Monotonicity]\label{SQM} Say that a DM exhibits \textbf{Status Quo Monotonicity} if for every choice set $S$ and $r,x \in S$, \[p_x(x,S)\geq p_r(x,S).\]
\end{definition}

While this condition is intuitive, it is not obvious if or when it holds in the RD-RAM. This is because when the status quo changes, say from $r$ to $x$, there are two effects. The first effect is that $x$ is now assured consideration, and this acts to increase the likelihood of choosing $x$. The second effect is that $x$ may now direct attention towards different alternatives. If the DM is now more likely to consider an alternative that dominates $x$, then this decreases the likelihood of choosing $x$. The total change depends on which of these two effects dominates.  To illustrate, we first establish by means of example that the IRA rule may violate \nameref{SQM}, and consequently so may the LRA and CRA rules. 

\begin{example}[SQM violation]
	\label{example: SQM_violation1} Suppose the DM's choices have an IRA representation with the preference $x\succ y\succ z$. If $1-\gamma_{y}(x) < 
\gamma_{z}(y)\left[1-\gamma_{z}(x)\right]$, then, $p_y(y,X) <  p_z(y,X)$.
\end{example}

 Note that since $\gamma_{z}(y) < 1$,  violations of \nameref{SQM} only arise when $\gamma_{y}(x) > \gamma_{z}(x)$ by a sufficient margin; $y$ substantially increases the chance of noticing $x$ (relative to $z$). This example suggests that \nameref{SQM} holds when attention is ``independent'' of the status quo. This is indeed the case when the DM admits an IRA (in fact, any CRA) representation, but it is not the case more generally.  
 
 To establish this, we first provide a formal definition of ``reference-independent'' attention for the IRA, LRA, and CRA models. The reference-independent versions of IRA and LRA are defined as in \autoref{attn:ira} and \autoref{attn:lra} respectively, where $\gamma$ and $\pi$ no longer depend on $r$. The reference-independent CRA model cannot simply be defined by \autoref{attn:cra} where $\pi'$ no longer depends on $r$, since then consideration set probabilities would not add up to one. Instead, we say that $\{\mu_r\}_{r\in X}$ has a reference-independent CRA representation if there exists a full support probability distribution $\pi$ on $\mathcal{X}$ such that for any $r\in D\subseteq S$, 
 \begin{equation*}
 \label{equation: ref-independent-CRA}
 \mu_r(D,S)=\sum_{D': D'\cap S=D}[\pi(D') + \pi(D'\setminus r)]
 \end{equation*}
 with the convention that $\pi(\emptyset) = 0$. Notice that for reference-independent CRA, if $x\neq r$, then 
 \[\sum_{D: x\in D}\mu_r(D,S)=\sum_{D: x\in D}\pi(D),\]
 which is independent from the reference alternative. That is, for any alternative that is distinct from the reference point, the consideration probability must be the same across different reference points and across different choice sets.\footnote{One can show that the reference-independent IRA is the intersection of the reference-independent LRA and CRA models. This parallels \autoref{proposition2} and ``verifies'' that our definition of reference-independent CRA is indeed the correct one.}

Now we can show that reference-independent CRA satisfies \nameref{SQM}, and consequently so does the reference-independent IRA. 
	
	\begin{proposition}
		\label{proposition: SQM - CRA}
		Suppose $\{p_r\}_{r \in X}$ admits a reference-independent CRA representation. Then $\{p_r\}_{r \in X}$ satisfies (strict) \nameref{SQM}.
		
	\end{proposition}

On the other hand, the reference-independent LRA need not satisfy \nameref{SQM}, as we demonstrate in the next example. 

\begin{example}[LRA violates SQM]
	\label{example: SQM_violation2}
	Suppose the DM's choices have a Luce Random Attention representation with the preference $x\succ y\succ z$. If $\pi$ is reference-independent, then \[p_z(y,X) = \frac{\pi(\{y,z\})}{\pi(\{y,z\}) + \pi(\{z\}) + \pi(\{x,z\}) + \pi(\{x,y,z\})}, \]
	and
\[p_y(y,X) = \frac{\pi(\{y,z\}) + \pi(\{y\})}{\pi(\{y,z\}) + \pi(\{y\}) + \pi(\{x,y\}) + \pi(\{x,y,z\}) }. \]

When $\pi(\{x,y\}) + \pi(\{x,y,z\})$ is large relative to $ \pi(\{y\})$, capturing positively correlated consideration between $x$ and $y$, we may have that $p_y(y,X) < p_z(y,X)$, which violates SQM.
\end{example}

In the other direction, imposing \nameref{SQM} directly onto the class of RD-RAM models restricts changes in $\mu_{r}$ across various $r$. In particular, the ``shifts'' in attention have to be relatively small.  It is plausible that tighter results may be derived with additional assumptions, either on the domain or on the attention process. We leave progress in this direction to future work.


\subsection{RD-RAM and Heterogeneity of Status Quo Bias}

 In \citet*{masatlioglu2014canonical}'s influential model of status quo bias, the DM is endowed with a psychological constraint function $Q:X\rightarrow \mathcal{X}$ such that for any choice problem $(S,r)$ the DM's deterministic choice is
$$c(S,r)=\argmax(\succeq, Q(r)\cap S).$$

Note that each consideration set $D \subseteq S$ may be written as $D=Q_i(r)\cap S$, where $Q_i:X\rightarrow \mathcal{X}$ is some psychological constraint function. Hence any RD-RAM may be viewed as a model of varying intensity of status quo bias, either due to individual variation (e.g., multiple selves or changing tastes) or population heterogeneity. Note that the RD-RAM allows for the distribution of psychological constraints, or types, to vary with the menu. 

Now consider the data from a fixed distribution of types as might arise from a DM with fixed distributions of selves or a fixed population. In either case, there is fixed distribution of psychological constraints, $Q_i: X\rightarrow \mathcal{X}$, where each constraint function is realized with probability $\nu_i$. Then the DM's (or population's) choices will be consistent with the Constant Random Attention model that satisfies \[\pi_r(D) = \sum_{i}\nu_i \mathbbm{1}(Q_i(r)=D).\]

\subsection{Unobservable Reference Points}

Our assumption that the reference is observable is plausible in many economic contexts. For example, in many cases it is natural that a status quo, past choice, or an outside option may serve as the reference.  When the reference is unobservable, then the assumptions made about the reference determine what can be inferred about the decision maker. 

If the reference is unobserved but fixed, then we may identify the set of potential references as the set of things chosen with positive probability and infer that everything not chosen is dominated by the reference. In this case we have only partial identification of preferences, but the model still provides some predictive power.  If we impose additional structure on the attention rule we may learn even more. For instance, suppose the attention rule coincides with the IRA model \citep{manzini2014stochastic}. In this case, the unobserved reference takes the role of the default, and preferences may be identified among those alternatives that dominate the reference. 

Alternatively, one might observe choice data for some fixed reference points but not all (the case of partial observability). If the attention rule coincides with the IRA, LRA or CRA models, one can uniquely infer preferences of the DM among those alternatives that dominate the reference points. One then needs to impose an assumption on observed choices to ensure that the revealed preference relation is consistent across different reference points. If this assumption is satisfied, our characterization result can still be employed as a test of observable implications of the models. 

If the reference is both unobserved and variable, we can make little inference about the decision maker's preferences absent restrictions on the reference selection process. Formally, suppose that given a menu $S$, each alternative is a potential reference.  The analyst then observes data of the following form: 
\[p_{\eta}(x,S)= \sum_{r \in S}p_{r}(x,S)\eta(r,S),\]
for $\eta:X \times \mathcal{X} \ra [0,1]$ and $\sum_{r \in S}\eta(r, S)=1$.  Here $\eta(r, S)$ is the probability that $r$ is the reference when choosing from menu $S$. When we allow for RD-RAM without full-support and an arbitrary $\eta$, then the model lacks testable implications.\footnote{When we restrict attention to RD-RAM with full support, the behavioral restrictions are exceedingly mild. Indeed, even if we suppose that $p_{r}$ follows the IRA rule, which is the most restrictive model that we consider, the RD-RAM with random reference may approximate any stochastic choice rule by setting $\eta(x,S)=p(x,S)$ and supposing $\gamma_r(x)=\epsilon$ for $x \neq r$ where $\epsilon > 0$ is sufficiently small.}

\begin{theorem}\label{theorem: unobservable}
	Every stochastic choice rule is an RD-RAM (without full support) with random reference.\end{theorem}

The essential piece of this argument is that we allow for unrestricted $\eta$. By imposing restrictions on how $\eta$ changes across menus, testable predictions may be obtained. For example, \citet{kibris2021random} assume that $\eta(r,S)$ satisfies the regularity property and show that under this assumption the preferences of the decision maker can be revealed from observed choices to a large extent. However, they assume that choice is deterministic once the reference point is fixed, and hence random choice is entirely due to randomness in reference points. An extension of RD-RAM that allows for unobserved reference points, where the reference rule satisfies certain desirable properties, is an interesting direction for further research. 

\subsection{Reference-Dependent Random Utility Models}

It is known that in a setting without a reference, the Constant Random Attention model (hence also IRA) is a special case of the Random Utility Model. Here we illustrate that the CRA model is a special case of the \textbf{Reference-Dependent Random Utility Model} (RD-RUM), which we now define. Let $\mathcal{R}$ denote the set of all possible linear orders on $X$. We say that $\{p_r\}_{r\in X}$ has an RD-RUM representation if for each reference point $r\in X$ there exists a probability distribution $\pi'_r$ on $\mathcal{R}$ such that
$$p_r(x,S)=\sum_{\succeq\in \mathcal{R}:\: x=\arg\max(\succeq,S)}\pi'_r(\succeq).$$
To see why the CRA model is a special case of the RD-RUM, suppose $\{p_r\}_{r\in X}$ has a CRA representation with $(\succeq, \{\pi_r\}_{r\in X})$, where $\pi_r$ is a probability distribution on $\mathcal{X}$. We can construct an RD-RUM representation of $\{p_r\}_{r\in X}$ as follows. First, for any linear order $\succeq'$ on $X$ and $D\subseteq X$, let $\succeq'_{D}$ denote the restriction of $\succeq'$ to $D$. Define the function $m_{\succeq}:\mathcal{X}\rightarrow \mathcal{R}$ by
$$m_{\succeq}(D)=\succeq' \text{ where } \succeq'_{D}=\succeq_{D},\: \succeq'_{X\setminus D}=\succeq_{X\setminus D}, \text{ and } x\succeq' y \text{ for any } x\in D, y\in X\setminus D.$$
In other words, $m_{\succeq}(D)$ and $\succeq$ completely agree on $D$ and $X\setminus D$, but the former ranks every element in $D$ as superior to elements in $X\setminus D$. 
Next, for any $\succeq'\in \mathcal{R}$, let 
$$\pi'_r(\succeq')=\sum_{D:\: m_{\succeq}(D)=\succeq'}\pi_r(D).$$ 
It is easy to see that $\{\pi'_r\}_{r\in X}$ defined as such is an RD-RUM representation of $\{p_r\}_{r\in X}$.

While the RD-RAM clearly has a non-trivial intersection with the RD-RUM, neither class of models nests the other. This can be easily  seen by recalling that the LRA model, which is nested in the RD-RAM, violates regularity. In the other direction, an RD-RUM with full support over $\mathcal{R}$ would violate \nameref{NCC}. We believe that study of the RD-RUM is a fruitful direction for future work.

\subsection{Related Literature}\label{section: lit}

Our paper is the first to study reference-dependence in a stochastic choice setting (as far as we are aware). Consequently, it is closely related to several strands of the choice theory literature. 

\textbf{Choice With Status Quo:} There is a vast literature on the effects of a status quo and a complete survey will not be attempted. Within the decision theoretic literature, the approach of assuming observability of the status quo has been followed in numerous other papers.  Early and influential work in this area includes \citet{masatlioglu2005rational} and \citet{masatlioglu2014canonical}, the latter of which which generalizes the former.  Additional papers with an observable status quo include \citet{Dean2008}, \cite{masatlioglu2013choice}, \cite{guney2018costly}, \cite{Maltz2020category}, and \cite{masatlioglu2021decision}. While some of these papers incorporate limited attention in a reference-dependent choice framework, all of these papers address only deterministic choice behavior.

\textbf{Choice With Limited Attention:} This paper is also related to the literature on imperfect attention and limited consideration.  In the deterministic setting, imperfect attention has been studied by \cite{masatlioglu2012revealed}, \cite{lleras2017more} and \citet{dean2017limited}. The paper by \citet{dean2017limited} is particularly relevant, as they also incorporate status quo.  They show that a deterministic model of limited attention with a status quo is able to explain choice reversals between non-status quo alternatives. A key difference between their approach and ours is that in their paper the attention process is independent of the status quo, whereas in our paper the status quo directly affects consideration probabilities.  

In the random choice literature, \citet{manzini2014stochastic},  \cite{brady2016menu}, and \cite{aguiar2017random} are the most closely related and have been discussed at length already. \cite{horan2019random} extends the analysis in \citet{manzini2014stochastic} by assuming that the no-choice behavior of the DM is unobservable.
Recently, \cite{cattaneo2020random} study a general random attention model which also features an unstructured attention process. A few recent papers study the case where both preferences and consideration sets are stochastic (see \cite{abaluck2017consumers, barseghyan2019heterogeneous,  dardanoni2020inferring, aguiar2021random, gibbard2021disentangling}).  Other stochastic choice papers which feature limited or imperfect attention include \cite{kovach2020}, which considers a satisficer with a random threshold, and \cite{kovach2021flm}, which studies menu-dependent focality with taste shocks. 
We view all of these papers as complementary, as we incorporate a reference and focus on reference-directed attention.

\textbf{Endogenous Reference Points:} Our paper is also related to the strand of literature that attempts to endogenize the reference point formation (\cite{koszegi2006model, bordalo2013salience, ok2014revealed, kibris2018theory, tserenjigmid2015choosing, lim2020, kibris2021random}). These papers mainly focus on situations where a DM's choices might be affected by a reference point that is not observed by the analyst. They are complementary to our study as we assume an observable status quo. An extension of our model that accommodates unobserved reference points is a fruitful direction for further research. 

\pagebreak
\appendix

\section{Proofs of Main Results}\label{section: proofs}

\subsection{Proof of \autoref{proposition2}}

First, suppose that $\mu_r$ has an IRA representation, i.e., \autoref{attn:ira} is satisfied for some $\gamma_r:X\rightarrow [0,1]$ such that $\gamma_r(r)=1$ and $\gamma_r(x)\in (0,1)$. Let $$\pi_r(D)=\prod_{x\in D}\gamma_r(x)\prod_{y\in X\setminus D}(1-\gamma_r(y)),$$ as in \cite{brady2016menu}. Then, if $r\in D\subseteq S$,
\begin{align*}
\frac{\pi_r(D)}{\sum_{D':\: r \in D'\subseteq S}\pi_r(D')}&=\frac{\prod_{x\in D}\gamma_r(x)\prod_{y\in X\setminus D}(1-\gamma_r(y))}{\sum_{D':\: r \in D'\subseteq S}\prod_{x\in D'}\gamma_r(x)\prod_{y\in X\setminus D'}(1-\gamma_r(y))}\\
&=\frac{\prod_{x\in D}\gamma_r(x)\prod_{y\in S\setminus D}(1-\gamma_r(y))}{\sum_{D':\: r \in D'\subseteq S}\prod_{x\in D'}\gamma_r(x)\prod_{y\in S\setminus D'}(1-\gamma_r(y))}\\
&=\prod_{x\in D}\gamma_r(x)\prod_{y\in S\setminus D}(1-\gamma_r(y))\\
&=\mu_r(D,S),
\end{align*}
where the first equality is by definition, the second equality follows from simple algebra, the third equality follows from the observation that the denominator is equal to 1, and the last equality follows from the hypothesis that $\mu_r$ has an IRA representation with this choice of $\gamma_r$. Hence, Independent Random Attention is a special case of Luce Random Attention. In addition,
\begin{align*}
\sum_{D': D'\cap S=D}\pi_r(D')&=\sum_{D': D'\cap S=D}\prod_{x\in D'}\gamma_r(x)\prod_{y\in X\setminus D'}(1-\gamma_r(y))\\
&=\prod_{x\in D}\gamma_r(x)\prod_{y\in S\setminus D}(1-\gamma_r(y))\Big[\sum_{D':D'\cap S=D}\prod_{x\in D'\setminus D}\gamma_r(x)\prod_{y\in (X\setminus S)\setminus D'}(1-\gamma_r(y))\Big]\\
&=\prod_{x\in D}\gamma_r(x)\prod_{y\in S\setminus D}(1-\gamma_r(y))\\
&=\mu_r(D,S),
\end{align*}
where the third equality follows from the observation that the term inside the big parentheses is equal to 1. Hence, Independent Random Attention is also a special case of Constant Random Attention. 
	
Now suppose $\mu_r$ has both LRA and CRA representations. That is, there exist $\pi_r$ satisfying \autoref{attn:lra} and $\pi'_r$ satisfying \autoref{attn:cra}. Let $\gamma_r(x)=\mu_r(\{r,x\},\{r,x\})$. Note that $\gamma_r(r)=1$ holds trivially and $\gamma_r(x)>0$ due to the full support assumption. It needs to be shown that $\gamma_r$ satisfies \autoref{attn:ira}. This is trivially satisfied when $S$ is binary. Suppose the claim holds whenever $S$ has a cardinality less than $k$, and let $S$ with $|S|=k>2$ be given. First, let $D\subsetneq S$ and pick $x\in S\setminus D$. Since $\mu_r$ has an LRA representation, it satisfies the IIA property for consideration sets. Hence, for any $D,D'\subseteq S\setminus x$ such that $r\in D\cap D'$, we have 
$$\frac{\mu_r(D,S)}{\mu_r(D,S\setminus x)}=\frac{\mu_r(D',S)}{\mu_r(D',S\setminus x)},$$
which implies that the above ratio only depends on $S$ and $x$. Let $\kappa(S,x)$ denote the above ratio so that 
$$\mu_r(D,S)=\mu_r(D,S\setminus x) \kappa(S,x).$$
for all $D\subseteq S\setminus x$. Adding over all $D\subseteq S\setminus x$, we get
$$\sum_{D:r\in D\subseteq S\setminus x} \mu_r(D,S)=\kappa(S,x) \Big(\sum_{D:r\in D\subseteq S\setminus x}\mu_r(D,S\setminus x)\Big)=\kappa(S,x),$$
where the second equality follows from $\sum_{D:r\in D\subseteq S\setminus x}\mu_r(D,S\setminus x)=1$. Now, plugging $\kappa(S,x)$ back in the above formula, we get
\begin{align*}
\mu_r(D,S)=\mu_r(D,S\setminus x)\sum_{D':\: r \in D'\subseteq S\setminus x}\mu_r(D',S).
\end{align*}
By induction hypothesis,
\begin{align*}
\mu_r(D,S\setminus x)=\prod_{y\in D}\gamma_r(y)\prod_{z\in S\setminus (D\cup x)}(1-\gamma_r(z)).
\end{align*}
Furthermore, since $\mu_r$ has a CRA representation, the probability that $x$ is not considered in $S$ is equal to the probability that $x$ is not considered in $\{r,x\}$. Hence, we have
\begin{align*}
\sum_{D':\: r \in D'\subseteq S\setminus x}\mu_r(D',S)=\mu_r(\{r\},\{r,x\})=(1-\gamma_r(x)),
\end{align*}
where the last equation holds by definition. Combining this equality with the previous ones, it follows that
\begin{align*}
\mu_r(D,S)=\prod_{y\in D}\gamma_r(y)\prod_{z\in S\setminus D}(1-\gamma_r(z))
\end{align*}
whenever $D\subsetneq S$. By using the fact that $\mu_r(S,S)=1-\sum_{D:\: r \in D\subsetneq S}\mu_r(D,S)$, the claim follows for all $D\subseteq S$. This concludes the proof of the proposition.

\subsection{Proof of \autoref{theorem:RDRAM}} 
Necessity is clear from the discussion in the main text. We prove sufficiency. For any $x,y\in X$, let $x\succeq y$ if $p_y(x,\{x,y\})>0$. By letting $S=T=\{x,y\}$ in the \nameref{SQA} axiom, we can see that for any $x,y\in X$, either $p_y(x,\{x,y\})>0$ or $p_x(y,\{x,y\})>0$ or both. Hence, $\succeq$ is complete. In addition, if $p_y(x,\{x,y\})>0$, then  \nameref{NCC} implies that $p_x(y,\{x,y\})=0$. Therefore, $\succeq$ is antisymmetric. To see that $\succeq$ is transitive, suppose $x\succeq y\succeq z$. If $x=y$ or $y=z$, then transitivity is obvious. Hence, suppose $x\succ y\succ z$. Then $p_y(x,\{x,y\})>0$ and $p_z(y,\{y,z\})>0$. \nameref{NCC} implies that $p_x(z,\{x,z\})=0$. By \nameref{SQA}, $p_z(x,\{x,z\})>0$, and hence $x\succ z$. 

Now for any $x$, let $D_x$ be the largest set in which $x$ is the $\succeq$-maximal element. For any $r\in D\subseteq S$, we have that $\{r,x\}\subseteq D\subseteq D_x\cap S$ for some $x\in S$. Next, note that either $x=r$ or $x\succ r$. If $x=r$, then $p_r(x,S)>0$ by \nameref{NRE}. On the other hand, if $x\neq r$, then $p_r(x,S)>0$ by \nameref{SQA} and \nameref{NCC}. Hence, we can choose $\mu_r(D,S)>0$ such that
$$p_r(x,S)=\sum_{D:\: \{r,x\}\subseteq D\subseteq D_x\cap S}\mu_r(D,S).$$
To illustrate, we can let $n_r(x,S) = |\{D\in \mathcal{X}| \: \{r,x\}\subseteq D\subseteq D_x\cap S\}|$ and $\mu_r(D,S)=p_r(x,S)/n_r(x,S)$ whenever $\{r,x\}\subseteq D\subseteq D_x\cap S$. It is easy to see that this results in the desired representation. The first paragraph also implies that $\succeq$ is unique. 

\subsection{Proof of \autoref{theorem:IRA}}

Necessity should be obvious from the discussion in the main text. Here we prove sufficiency. Let $x\succeq y$ whenever $p_y(x,\{x,y\})>0$. Since $\{p_r\}_{r\in X}$ is an RD-RAM, $\succeq$ is a linear order on $X$. Recall that $x$ is the dominant alternative in $S$ if and only if $p_y(x,S)>0$ for any $y\in S$. Hence, by our definition of $\succeq$, $x$ is the dominant alternative in $S$ if and only if $x=\max(\succeq, S)$. 

\begin{lemma}
	If $\{p_r\}_{r\in X}$ satisfies \nameref{IDA} and \nameref{RIDA}, then for any $(S,r)$ and $y\in S\setminus x$ where $x$ is the dominant alternative in $S$,
	$$p_r(y,S)=p_r(y,S\setminus x)p_r(r,\{x,r\}).$$
\end{lemma}

\begin{proof}
	RIDA implies that there exists a constant $\kappa(S)$ such that $p_r(y,S)=\kappa(S)p_r(y,S\setminus x)$ for all $y\in S\setminus x$. By adding over all $y\in S\setminus x$, we get that $\sum_{y\in S\setminus x}p_r(y,S)=\kappa(S)$, or alternatively, $1-p_r(x,S)=\kappa(S)$. Since $x$ is the dominant alternative in $S$, IDA implies that $1-p_r(x,S)=1-p_r(x,\{r,x\})=p_r(r,\{x,r\})$. Hence, \[p_r(y,S)=p_r(y,S\setminus x)\kappa(S)=p_r(y,S\setminus x)(1-p_r(x,S))= p_r(y,S\setminus x)p_r(r,\{x,r\}),\]
	as desired.
\end{proof}

The next step is to define $\gamma_r(x)$ for any $r,x\in X$. Let 
$$\gamma_r(x)=\max\{p_r(x,\{r,x\}),p_x(r,\{r,x\})\}.$$
If $r=x$, then we have $\gamma_r(r)=1$. Moreover, since $\{p_r\}_{r\in X}$ is an RD-RAM, if $r\neq x$, then one of the terms inside the curly brackets is zero and the other one is strictly between zero and one due to \nameref{NCC} and \nameref{SQA}, so that $\gamma_r(x)\in (0,1)$. For any $(S,r)$, the consideration set probabilities are defined by 
$$\mu_r(D,S)=\prod_{x\in D}\gamma_r(x)\prod_{y\in S\setminus D}(1-\gamma_r(y)).$$

\begin{lemma}
	For any $(S,r)$, $x\in S$, and $D\subseteq S\setminus x$, 
	$$\mu_r(D,S)=\mu_r(D,S\setminus x)\mu_r(\{r\},\{r,x\}).$$
\end{lemma}

\begin{proof}
	This claim follows from simple algebra and definitions:
	\begin{align*}
	\mu_r(D,S) &= \prod_{y\in D}\gamma_r(y)\prod_{z\in S\setminus D}(1-\gamma_r(z))=\prod_{y\in D}\gamma_r(y)\prod_{z\in S\setminus (D\cup x)}(1-\gamma_r(z))(1-\gamma_r(x))\\
	&= \mu_r(D,S\setminus x)\mu_r(\{r\},\{r,x\}).
	\end{align*}
\end{proof}

\begin{lemma}
	For any $(S,r)$ and $x\in S$,
	\[p_r(x,S)=\sum_{D\subseteq S:\: x = \arg\max(\succeq, D)}\mu_r(D,S).\]
\end{lemma}

\begin{proof}
	The proof is by induction. First, suppose $S=\{x,r\}$. If $p_r(x,S)=0$, then $r\succ x$  and we are done. If $p_r(x,S)>0$, then $x\succ r$ and 
	$$p_r(x,\{x,r\}) = \max\{p_r(x,\{x,r\}),p_x(r,\{x,r\})\}  = \gamma_r(x)=\mu_r(\{x,r\},\{x,r\}),$$ as desired.
	
	Now suppose the claim holds for all choice sets with cardinality less than $k$. Let $|S|=k$ and $r\in S$ be given. Let $y = \max(\succeq, S)$. If $y=r$, then the claim is trivial, since we have $p_r(r,S)=1$ and $p_r(x,S)=0$ for all $x\in S\setminus r$. Hence, suppose $y\neq r$ and let $x\neq y$ be given. Then,
	\begin{align*}
	p_r(x,S)&=p_r(x,S\setminus y)p_r(r,\{r,y\})\\
	&=\Big(\sum_{D\subseteq S\setminus y:\: x = \arg\max(\succeq, D)}\mu_r(D,S\setminus y)\Big)\mu_r(r,\{r,y\})\\
	&=\sum_{D\subseteq S\setminus y:\: x = \arg\max(\succeq, D)}\mu_r(D,S)\\
	&=\sum_{D\subseteq S:\: x = \arg \max(\succeq, D)}\mu_r(D,S),
	\end{align*}
	where the first equality follows from the first lemma, the second equality follows from induction, the third equality follows from the previous lemma, and the last equality follows from the fact that $y\succ x$. This shows that the representation holds for all $x\in S\setminus y$. Finally, since $p_r(y,S)=1-\sum_{x \in S\setminus y}p_r(x,S)$, the representation also holds for $y$. 
\end{proof}

\subsection{Proof of \autoref{theorem:logit}}

\subsubsection{Necessity}

The necessity of \nameref{RIDA} is obvious from the main text. Here we show the necessity of \nameref{DORA}.

\begin{lemma}
	If $\{p_r\}_{r\in X}$ has a Luce Random Attention representation $(\succeq, \{\pi_r\}_{r\in X})$, then it satisfies \nameref{DORA}.
\end{lemma}
\begin{proof}
	Let $(S,r)$ be given. Recall that $P_r$ is the set of all alternatives which dominate $r$ in a binary comparison (i.e., $s\in P_r \Leftrightarrow  p_r(s,\{r,s\})>0$ and $s\neq r$). From the representation, 
	$$\mathcal{O}^r_S=\frac{1-p_r(r,S)}{p_r(r,S)}=\frac{\sum_{D:\: r\in D\subseteq S}\pi_r(D)}{\sum_{D:\: r\in D\subseteq S\setminus P_r}\pi_r(D)}-1.$$
	Hence,  for any $U\subseteq P_r$ such that $U\cap S\neq \emptyset$,
	\begin{align*}
	\Delta_{U}\mathcal{O}^r_S&=\frac{1-p_r(r,S)}{p_r(r,S)}-\frac{1-p_r(r,S\setminus U)}{p_r(r,S\setminus U)}=\frac{\sum_{D:\: r\in D\subseteq S}\pi_r(D)-\sum_{D:\: r\in D\subseteq S\setminus U}\pi_r(D)}{\sum_{D:\: r\in D\subseteq S\setminus P_r}\pi_r(D)}\\&=\frac{\sum_{D:\: r\in D\subseteq S \text{ and } D\cap U\neq \emptyset}\pi_r(D)}{\sum_{D:\: r\in D\subseteq S\setminus P_r}\pi_r(D)},
	\end{align*}
	where we use $S\setminus P_r  = (S\setminus U)\setminus P_r$. Similarly, for $\mathcal{U}=\{U_1,\dots, U_n\}$ such that $U_i\subseteq P_r$ and $S\cap U_i\neq \emptyset$ for all $i\in \{1,\dots, n\}$,
	\begin{align*}
	\Delta_{\mathcal{U}}\mathcal{O}^r_S&=\Delta_{\mathcal{U}\setminus U_1}\mathcal{O}^r_S-\Delta_{\mathcal{U}\setminus U_1}\mathcal{O}^r_{S\setminus U_1}=\frac{\sum_{D:\: r\in D\subseteq S \text{ and for all } i, \: D\cap U_i\neq \emptyset}\pi_r(D)}{\sum_{D:\: r\in D\subseteq S\setminus P_r}\pi_r(D)}\\
	&\geq \frac{\pi_r(S)}{\sum_{D:\: r\in D\subseteq S\setminus P_r}\pi_r(D)}>0,
	\end{align*}
	as desired.
\end{proof}

\subsubsection{Sufficiency}
First, define $x\succeq y$ if $p_y(x,\{x,y\})>0$. As before, $\succeq$ is a linear order on $X$. We next construct $\{\pi_r\}_{r\in X}$ from $\{p_r\}_{r\in X}$. The following theorem will be useful for this purpose.

\begin{theorem}[M\"obius Inversion \citep*{shafer1976mathematical}]
	\label{mobius}
	If $\Theta$ is a finite set and $f$ and $g$ are functions on $2^{\Theta}$, then
	$$f(A)=\sum_{B\subseteq A}g(B)$$
	for all $A\subseteq \Theta$ if and only if 
	$$g(A)=\sum_{B\subseteq A}(-1)^{|A\setminus B|}f(B)$$
	for all $A\subseteq \Theta$. 
\end{theorem}

Let $r\in X$ be given, and denote by $\mathcal{D}_r$ the collection of subsets of $X$ which have $r$ as the dominant alternative. To construct $\{\pi_r\}_{r\in X}$, we first construct weights $\{\lambda_r\}_{r\in X}$ such that the former is the scale normalization of the latter. First, let $\lambda_r(\{r\})=1$. For $D\in \mathcal{D}_r$ such that $D\neq \{r\}$, $\lambda_r(D)$ is a positive constant to be determined later. Next for each $T\in \mathcal{D}_r$, let 
$$\lambda_r^T(T)=\sum_{D:\: r\in D\subseteq T}\lambda_r(D).$$ 
Similarly, for any $S\subseteq X$ such that $T\subseteq S\subseteq T\cup P_r$, where $P_r$ is the set of all alternatives that dominate $r$, the term $\lambda_r^T(S)$ represents
$$\lambda_r^T(S)=\sum_{D:\: r\in D\subseteq T}\lambda_r(D\cup (S\cap P_r)).$$ 
It turns out that once we have chosen $\lambda_r(D)$ in the appropriate way for $D\in \mathcal{D}_r$, for choice data to be consistent with the LRA model, $\lambda_r^T(S)$ must be given by the following equation:
\begin{equation}\label{equation: mobius}
\lambda_r^T(S)=\lambda_r^T(T)\sum_{D:\: T\subseteq D\subseteq S}(-1)^{|S\setminus D|}\frac{1}{p_r(r,D)}.
\end{equation}
Hence, we first choose $\lambda_r(D)$ for $D\in \mathcal{D}_r$, and then use Equation \ref{equation: mobius} to define $\lambda_r(\cdot)$ for all other choice sets. For this construction to work, we first need to show that if $\{p_r\}_{r\in X}$ satisfies \autoref{DORA}, then $\lambda_r^T(S)>0$ as long as $\lambda_r^T(T)>0$. 

\begin{lemma}
	\label{lemma: positivity}
	Assume that $\{p_r\}_{r\in X}$ is an RD-RAM. If $\{p_r\}_{r\in X}$ satisfies \nameref{DORA}, then for any $T\in \mathcal{D}_r$ and $S\subseteq X$ such that $T\subseteq S\subseteq T\cup P_r$, 
	$$\alpha_r^T(S)=\sum_{D:\: T\subseteq D\subseteq S}(-1)^{|S\setminus D|}\frac{1}{p_r(r,D)}>0.$$
\end{lemma}
\begin{proof}
	Let $T\in \mathcal{D}_r$ be given, and define $\alpha_r^T(S)$ as above for any $S\subseteq X$ with $T\subseteq S\subseteq T\cup P_r$. We will invoke M\"obius Inversion to get an expression for $p_r(r,S)$. To this end, let $\Theta = P_r$. Then, $\alpha_r^T(T\cup \cdot )$ and $p_r(r,T\cup \cdot)$ are functions on $2^{\Theta}$ such that for any $A\subseteq \Theta$,
	$$\alpha_r^T(T\cup A) =  \sum_{B:\: B\subseteq A}(-1)^{|A\setminus B|}\frac{1}{p_r(r,T\cup B)}$$
	Hence, by M\"obius Inversion, 
	$$\frac{1}{p_r(r,T\cup A)}=\sum_{B:B\subseteq A} \alpha_r^T(T\cup B)$$
	or, alternatively,
	$$p_r(r,T\cup A)=\frac{1}{\sum_{B:B\subseteq A} \alpha_r^T(T\cup B)}.$$
	Now, letting $S=T\cup A$ and $D=T\cup B$ in the above equation, we get
	$$p_r(r,S)=\frac{1}{\sum_{D: T\subseteq D\subseteq S}\alpha_r^T(D)}.$$
	Now, note that 
	$$\mathcal{O}^r_S=\frac{1-p_r(r,S)}{p_r(r,S)}=\frac{1}{p_r(r,S)}-1 =\sum_{D: T\subseteq D\subseteq S}\alpha_r^T(D)-1.$$
	In addition, for any $U\subseteq P_r$ such that $U\cap S\neq \emptyset$, we have $U\cap T=\emptyset$. Therefore,
	\begin{align*}
	\Delta_{U}\mathcal{O}^r_S&=\frac{1-p_r(r,S)}{p_r(r,S)}-\frac{1-p_r(r,S\setminus U)}{p_r(r,S\setminus U)}=\sum_{D: T\subseteq D\subseteq S}\alpha_r^T(D)-\sum_{D: T\subseteq D\subseteq S\setminus U}\alpha_r^T(D)\\&=\sum_{D: T\subseteq D\subseteq S \text{ and } D\cap U\neq \emptyset}\alpha_r^T(D).
	\end{align*}
	Now, let $\mathcal{U}=\{\{x_i\}| \: x_i\in S\setminus T\}$. Note that $x_i\in P_r$ as $S\setminus T\subseteq P_r$. Then, using the above result, we get
	$$\Delta_{\mathcal{U}}\mathcal{O}^r_S=\sum_{D:\: T\subseteq D\subseteq S \text{ and for all } \{x_i\}\in \mathcal{U},\: D\cap \{x_i\}\neq \emptyset}\alpha_r^T(D)=\alpha_r^T(S).$$
	By \nameref{DORA}, this term is strictly positive.
\end{proof}

Since $\lambda_r^{\{r\}}(\{r\})=1>0$, \autoref{lemma: positivity} implies that for any $S$ such that $S\setminus P_r=\{r\}$, $\lambda_r^{\{r\}}(S)>0$. We let $\lambda_r(S)=\lambda_r^{\{r\}}(S)$ for any such $S$. Suppose we have defined $\lambda_r(S)>0$ for all $S$ with $|S\setminus P_r|<k$ where $k>1$. Let $T\in \mathcal{D}_r$ be such that $|T|=k$. For any $S\subseteq X$ such that $T\subseteq S\subseteq T\cup P_r$, we let $\lambda_r(S)$ satisfy 
$$\lambda_r(S)=\lambda_r^{T}(S)-\sum_{D:\: r\in D\subsetneq T}\lambda_r(D\cup (S\cap P_r)).$$
Note that by induction $\lambda_r(D\cup(S\cap P_r))$ is assumed to be defined for all $D\subsetneq T$ in the above equation. Similarly, $\lambda_r(D)$ is defined for all $D\subsetneq T$. Therefore, the only undefined term in $\lambda_r^T(T)  = \sum_{D:\: r\in D\subseteq T}\lambda_r(D)$ is $\lambda_r(T)$. Now, since $\lambda_r^T(S)$ is given by Equation \ref{equation: mobius}, by using \autoref{lemma: positivity}, we choose $\lambda_r(T)$ sufficiently large to ensure that 
\begin{align*}
\lambda_r^{T}(S)&=\lambda_r^{T}(T)\sum_{D:\: T\subseteq D\subseteq S}(-1)^{|S\setminus D|}\frac{1}{p_r(r,D)}\\
&=\Big(\sum_{D:\: r\in D\subseteq T}\lambda_r(D)\Big)\Big(\sum_{D:\: T\subseteq D\subseteq S}(-1)^{|S\setminus D|}\frac{1}{p_r(r,D)}\Big)\\
&> \sum_{D:\: r\in D\subsetneq S\setminus P_r}\lambda_r(D\cup (S\cap P_r)),
\end{align*}
for all $S$ with $T\subseteq S\subseteq T\cup P_r$. Hence, we have now defined $\lambda_r(S)>0$ for all $S$ with $|S\setminus P_r|=k$. We can follow this construction to define $\lambda_r(S)$ for all $S\subseteq X$. Lastly, note that for any $S$ with $S\setminus P_r = T$,
	\begin{equation}\label{equation: lemma5}
	\lambda_r^T(S) =\lambda_r(S)+\sum_{r\in D\subsetneq T}\lambda_r(D\cup (S\setminus T))=\sum_{D:\: (S\setminus T)\cup r\subseteq D\subseteq S}\lambda_r(D).
	\end{equation}

\begin{lemma}
	\label{lemma: luce-characterization}
	Assume that $\{p_r\}_{r\in X}$ is an RD-RAM. If $\{p_r\}_{r\in X}$ satisfies \nameref{RIDA} and \nameref{DORA}, then for any $(S,r)$ and $x\in S$, 
	$$p_r(x,S)=\sum_{D\subseteq S:\: x = \arg\max(\succeq, D)}\frac{\lambda_r(D)}{\sum_{D': \: r\in D'\subseteq S}\lambda_r(D')}.$$
\end{lemma}

\begin{proof}
	Let $r\in X$ be given. First, we show that the representation holds for binary choice sets. If $p_r(x,\{r,x\})=0$, then $r\succ x$ and the representation holds trivially. If $p_r(x,\{r,x\})>0$, then by construction $\lambda_r(\{r\})=1$ and $\lambda_r(\{r,x\})=\frac{1}{p_r(r,\{r,x\})}-1$. Hence,
	$$\frac{\lambda_r(\{r,x\})}{\lambda_r(\{r\})+\lambda_r(\{r,x\})}=1-p_r(r,\{r,x\})=p_r(x,\{r,x\}).$$
	
	The next step is to show that the claim holds for all $S$ whenever $x=r$. Let $T=S\setminus P_r$. By applying M\"obius Inversion to \autoref{equation: mobius}, we get
	$$\frac{\lambda_r^T(T)}{p_r(r,S)}=\sum_{D:\: T\subseteq D\subseteq S}\lambda_r^T(D),$$
	and hence
	\begin{align*}
	p_r(r,S) &=\frac{\lambda_r^T(T)}{\sum_{D:\: T\subseteq D\subseteq S}\lambda_r^T(D)} = \frac{\sum_{D':\: r\in D'\subseteq T}\lambda_r(D')}{\sum_{D:\: T\subseteq D\subseteq S}\sum_{D'': \: (D\setminus T)\cup r\subseteq D''\subseteq D}\lambda_r(D'')}\\
	&=\frac{\sum_{D': r\in D'\subseteq T}\lambda_r(D')}{\sum_{D'': r\in D''\subseteq S}\lambda_r(D'')}
	\end{align*}
	where the second equality follows from \autoref{equation: lemma5}. Since $r$ is the $\succeq$-best alternative only in subsets of $T$, this proves the claim for $x=r$. 
	
	Now suppose the original claim holds for all choice sets with cardinality less than $k$, and let $S$ with $|S|=k$ be given. For any $x$ that is not chosen with positive probability, the claim is trivial. If there is only one alternative other than $r$, say $x$, that is chosen with positive probability, then the representation must hold, since $p_r(x,S)=1-p_r(r,S)$ and the representation holds for the reference alternative.  Hence, suppose there are at least two alternatives other than $r$ chosen with positive probability. Let $z$ denote the dominant alternative in $S$. If $x$ is chosen with positive probabilty, then \nameref{RIDA} implies that 
	$$\frac{p_r(x,S)}{p_r(r,S)}=\frac{p_r(x,S\setminus z)}{p_r(r,S\setminus z)}.$$
	Hence, by using the induction argument, the fact that $z$ is the dominant alternative in $S$, and that the representation holds whenever $x=r$, we get
	\begin{align*}p_r(x,S)&=\frac{p_r(x,S\setminus z)}{p_r(r,S\setminus z)}p_r(r,S)=\frac{\frac{\sum_{D\subseteq S\setminus z:\: x=\arg\max(\succeq, D)}\lambda_r(D)}{\sum_{D'\subseteq S\setminus z}\lambda_r(D')}}{\frac{\sum_{D\subseteq S\setminus z:\: r = \arg\max(\succeq ,D)}\lambda_r(D)}{\sum_{D'\subseteq S\setminus z}\lambda_r(D')}}\frac{\sum_{D\subseteq S:\: r=\arg\max(\succeq, D)}\lambda_r(D)}{\sum_{D'\subseteq S}\lambda_r(D')}\\
	&=\frac{\sum_{D\subseteq S:\: x=\arg\max(\succeq, D)}\lambda_r(D)}{\sum_{D\subseteq S:\: r = \arg\max(\succeq, D)}\lambda_r(D)}\frac{\sum_{D\subseteq S:\: r=\arg\max(\succeq, D)}\lambda_r(D)}{\sum_{D'\subseteq S}\lambda_r(D')}\\
	&=\sum_{D\subseteq S:\: x=\arg\max(\succeq, D)}\frac{\lambda_r(D)}{\sum_{D'\subseteq S}\lambda_r(D')},
	\end{align*}
	as desired. Hence, we have shown that the representation must hold for all $x\in S\setminus z$. Lastly, the claim holds for $z$ as well, since $p_r(z,S)=1-\sum_{x\in S\setminus z}p_r(x,S)$. 
\end{proof}

To conclude the proof of the theorem, for any $r\in X$, we define a probability distribution $\pi_r$ given by
\begin{equation}
\label{equation: luce_weights}
\pi_r(S)=\frac{\lambda_r(S)}{\sum_{D:\: r\in D\subseteq X} \lambda_r(D)}.
\end{equation}
\autoref{lemma: luce-characterization} guarantees that $(\succeq,\{\pi_r\}_{r\in X})$ represents $\{p_r\}_{r\in X}$. 

\subsection{Proof of \autoref{theorem:CRA}}

\subsubsection{Necessity}

The necessity of \autoref{IDA} is obvious from the main text. Here we show the necessity of \autoref{DPCRA}.

\begin{lemma}
	\label{lemma:CRA-necessity}
	If $\{p_r\}_{r\in X}$ has a Constant Random Attention representation $(\succeq, \{\pi_r\}_{r\in X})$, then it satisfies \nameref{DPCRA}.
\end{lemma}

\begin{proof}
Let $D_r$ be the largest set in which $r$ is the dominant alternative, i.e., $D_r=X\setminus P_r$. First, notice that for any $(S,r)$ and any $U\subseteq P_r$ such that $S\cap U\neq \emptyset$, 
\begin{align*}
\Delta_{U}p_r(r,S)&=p_r(r,S\setminus U)-p_r(r,S)=\sum_{D:\: r\in  D\subseteq (X\setminus (S\setminus U))\cup D_r}\pi_r(D)-\sum_{D:\: r\in  D\subseteq (X\setminus S)\cup D_r}\pi_r(D)\\
&=\sum_{D:\: r\in  D\subseteq (X\setminus (S\setminus U))\cup D_r \text{ and } D\cap (S\cap U)\neq \emptyset}\pi_r(D)\geq \pi_r((S\cap U)\cup r)>0.
\end{align*}
Similarly, for any collection $\mathcal{U}=\{U_1,\dots, U_n\}$ such that $U_i\subseteq P_r$ and $S\cap U_i\neq \emptyset$ for every $i\in \{1,\dots, n\}$,
\begin{align*}
\Delta_{\mathcal{U}}p_r(r,S)=&\sum_{D:\: r\in  D\subseteq [X\setminus (S\setminus (\bigcup_{U_i\in \mathcal{U}}U_i))]\cup D_r \text{ and for all }i, \: D\cap (S\cap U_i)\neq \emptyset}\pi_r(D)\\
&\geq \pi_r((S\cap (\cup_{U_i\in \mathcal{U}}U_i))\cup r)>0.
\end{align*}
\end{proof}

\subsubsection{Sufficiency}

Now, we show that \autoref{IDA} and \autoref{DPCRA} are sufficient for a CRA representation. As before, we let $x\succeq y$ if  $p_y(x,\{x,y\})>0$, which is a linear order since $\{p_r\}_{r\in X}$ is an RD-RAM. Let $r\in X$ be given. Let $D_r$ be the largest set in which $r$ is the dominant alternative, i.e., $D_r=X\setminus P_r$. For any $S\supseteq D_r$, let 
\begin{equation}
\label{equation: CRA-lambdas}
\lambda_r(S)=\sum_{D:\: D_r\subseteq D\subseteq S}(-1)^{|S\setminus D|}p_r(r,(X\setminus D)\cup D_r).
\end{equation}
The next lemma shows that if $\{p_r\}_{r\in X}$ satisfies \autoref{DPCRA}, then $\lambda_r(S)>0$ for any choice set $S\subseteq X$ such that $S\supseteq D_r$.  
\begin{lemma}
\label{lemma: CRA-lambdas}
Assume that $\{p_r\}_{r\in X}$ is an RD-RAM and it satisfies \autoref{DPCRA}. Then $\lambda_r(S)$ defined by \autoref{equation: CRA-lambdas} is strictly positive for any choice set $S\subseteq X$ such that $S\supseteq D_r$.  
\end{lemma}
\begin{proof}
Let $r\in X$ be given and for any $S\supseteq D_r$, define $\lambda_r(S)$ as above. Given the definition of $\lambda_r$, M\"obius Inversion implies that for any $S'\supseteq D_r$,
$$p_r(r,(X\setminus S')\cup D_r)=\sum_{D:\: D_r\subseteq D\subseteq S'}\lambda_r(D).$$
Letting $S=(X\setminus S')\cup D_r$ and noting that $S' = (X\setminus S)\cup D_r$, we get
\begin{equation}
\label{equation: CRA-reference_prob}
p_r(r,S)=\sum_{D:\: D_r\subseteq D\subseteq (X\setminus S)\cup D_r}\lambda_r(D).
\end{equation}
This shows that for any $S\supseteq D_r$, $p_r(r,S)$ is given by \autoref{equation: CRA-reference_prob}. Now let $S\supseteq D_r$ be given, and let $\mathcal{U}=\{\{x_i\}| \: x_i\in S\cap P_r\}$. By \autoref{DPCRA}, $\Delta_{\mathcal{U}}p_r(r,X)>0$. Notice that 
\begin{align*}
0<\Delta_{\mathcal{U}}p_r(r,X)=\sum_{D:\: D_r\subseteq  D\subseteq [X\setminus (X\setminus (\bigcup_{U_i\in \mathcal{U}}U_i))]\cup D_r \text{ and for all } i,\: D\cap (X\cap U_i)\neq \emptyset}\lambda_r(D)=\lambda_r(S)
\end{align*}
as desired.
\end{proof}

By \autoref{equation: CRA-reference_prob},
$$p_r(r,D_r)=\sum_{D:\: D_r\subseteq D\subseteq X}\lambda_r(D).$$
Since, by definition, $p_r(r,D_r)=1$, this together with \autoref{lemma: CRA-lambdas} imply that $\lambda_r$ is a probability distribution on sets $S\supseteq D_r$. Now choose any $\pi_r$ over $\mathcal{X}$ such that $\pi_r(D)>0$ for each $D\in \mathcal{X}$ and for any $S\supseteq D_r$,
\begin{equation}
\label{equation: CRA-reference_prob2}
\lambda_r(S)=\sum_{D:\: (S\setminus D_r)\cup r \subseteq D\subseteq S}\pi_r(D).
\end{equation}

\begin{lemma}
Suppose $\{p_r\}_{r\in X}$ is an RD-RAM and it satisfies \nameref{IDA} and \nameref{DPCRA}. Then for any $(S,r)$ and $x\in S$, 
$$p_r(x,S)=\sum_{D:\: x = \arg\max(\succeq, D\cap S)}\pi_r(D).$$
\end{lemma}
\begin{proof}
First, notice that by \autoref{equation: CRA-reference_prob} and \autoref{equation: CRA-reference_prob2} for any $S\supseteq D_r$, 
$$p_r(r,S)=\sum_{D:\: D_r\subseteq D\subseteq (X\setminus S)\cup D_r}\lambda_r(D)=\sum_{D:\: r \in D\subseteq (X\setminus S)\cup D_r}\pi_r(D)=\sum_{D:\: r\in D\cap S\subseteq D_r}\pi_r(D).$$
On the other hand, if $S\not \supseteq D_r$, then Irrelevance of Dominated Alternatives implies that $$p_r(r,S)=p_r(r,S\cup D_r)=\sum_{D:\: r\in D\cap (S\cup D_r)\subseteq D_r}\pi_r(D)=\sum_{D:\: r\in D\cap S\subseteq D_r}\pi_r(D).$$ 
Tihs shows that the representation holds for the reference alternative. This also shows that the representation holds for binary choice sets $\{x,r\}$, since it holds for $p_r(r,\{x,r\})$ and $p_r(x,\{x,r\})=1-p_r(r,\{x,r\})$. Suppose the representation holds for all choice sets with cardinality less than $k$. Let $S$ with $|S|=k$ be given. If $S$ contains only one alternative, say $x$, that dominates $r$, then since $|S|\geq 3$, there exists $z\in S$ such that $p_r(z,S)=0$. \nameref{IDA} implies that $p_r(x,S)=p_r(x,S\setminus z)$ and $p_r(r,S)=p_r(r,S\setminus z)$. Since by induction the representation holds for $S\setminus z$, it also holds for $S$. Now suppose there are at least two alternatives which dominate $r$. Let $z\in S$ be an alternative that dominates $r$ but is dominated by all other alternatives $x$ for which $p_r(x,S)>0$. Let $x\notin \{r,z\}$ with $p_r(x,S)>0$ be given. By \nameref{IDA}, $p_r(x,S)=p_r(x,S\setminus z).$
By induction argument,
$$p_r(x,S\setminus z)=\sum_{D:\: x = \arg\max(\succeq, D\cap (S\setminus z))}\pi_r(D).$$
But since $x\succ z$, $x$ is $\succeq$-best in $D\cap S$ if and only if it is $\succeq$-best in $D\cap (S\setminus z)$. Hence, 
$$p_r(x,S)=p_r(x,S\setminus z)=\sum_{D:\: x = \arg\max(\succeq, D\cap S)}\pi_r(D).$$
Since the representation holds for the reference alternative and all alternatives that dominate $z$, to conclude the proof, we only need to show that the representation also holds for $z$. This follows from the fact $p_r(z,S)=1-\sum_{x \in S\setminus z}p_r(x,S)$.
\end{proof}

\subsection{Proof of \autoref{proposition: SQM - CRA}}

\begin{proof}
	Let $(\succeq , \pi)$ denote the reference-independent CRA representation of $\{p_r\}_{r \in X}$. If $r\succ x$, then $p_r(x,S) = 0 <p_x(x,S)$, and hence \nameref{SQM} is satisfied. If $x\succ r$, notice that
	\begin{align*}
	p_r(x,S) &= \sum_{D\ni r:\: x = \arg\max(\succeq, D\cap S)} [\pi(D) + \pi(D\setminus r)]\\
	&= \sum_{D :\: x = \arg\max(\succeq, D\cap S)} \pi(D) \\
	& < \sum_{D :\: x = \arg\max(\succeq, D\cap S)} [\pi(D) +\pi(D\setminus x)]\\
	& = p_x(x,S),
	\end{align*}
	and (strict) \nameref{SQM} is again satisfied. 
\end{proof}

\subsection{Proof of \autoref{theorem: unobservable}}

\begin{proof} Let $p$ denote any stochastic choice rule. We construct an RD-RAM representation (without full support) with random reference. Let $\succeq$ be any linear order. Fix $S$ and for each $r \in S$, we construct consideration rule that always leads to the choice of $r$. That is, if $D \cap U_\succeq(r) =\{r\}$, let $\mu_r(D,S) \in (0,1)$ be arbitrary and let $\mu_r(D,S) =0$ otherwise. Finally, for every $S$ and each $x \in S$, let $\eta(x,S)=p(x,S)$. Then for every $r$,  $p_r(r, S)=1$. Hence, for any $x \in S$, 
\[p_{\eta}(x,S)=\sum_{r \in S}p_r(x, S)\eta(r, S)=\eta(x,S)=p(x,S).\]
\end{proof}

\pagebreak

\bibliographystyle{abbrvnat}
\bibliography{rdra.bib}

\end{document}